\documentclass{IEEEtran}

\newcommand{\cJ}{\mathcal J}
\newcommand{\cK}{\mathcal K}

\newcommand{\cN}{\mathcal N}

\newcommand{\cP}{\mathcal P}

\newcommand{\cS}{\mathcal S}

\newcommand{\bbZ}{\mathbb{Z}}

\newcommand{\vu}{\textbf{u}}

\newcommand{\vx}{\textbf{x}}

\newcommand{\vz}{\bm{z}}

\newcommand{\bmat}{\begin{bmatrix}}
\newcommand{\emat}{\end{bmatrix}}

\newcommand{\bsmat}{\begin{bsmallmatrix}}
\newcommand{\esmat}{\end{bsmallmatrix}}

\newcommand{\iter}[2]{{#1}^{(#2)}}
\newcommand{\st}{\text{s.t.}}

\newcommand{\revised}[1]{\textcolor{black}{#1}}
\newcommand{\secrevised}[1]{\textcolor{black}{#1}}

\providecommand{\norm}[1]{\left\lVert#1\right\rVert}

\usepackage{bm}
\usepackage{commath}
\usepackage{comment}
\usepackage{algorithm}
\usepackage[noend]{algpseudocode}
\usepackage{amsmath,amsthm,amssymb,amsfonts}

\usepackage{float}
\usepackage{subfig}
\usepackage{xcolor}
\usepackage{dcolumn}
\usepackage{caption}
\usepackage{makecell}
\usepackage{multirow}
\usepackage{graphicx}
\usepackage{tabularray}
\usepackage{stackengine}
\usepackage{url}
\usepackage{cite}
\usepackage{bigfoot}
\usepackage{textcomp}
\usepackage{enumitem}
\usepackage{anyfontsize}
\usepackage[euler]{textgreek}
\usepackage[hidelinks]{hyperref}
\DeclareNewFootnote[para]{default}

\newtheorem{remark}{Remark}
\newtheorem{theorem}{Theorem}

\newtheorem{assumption}{Assumption}
\newtheorem{proposition}{Proposition}

\def\BibTeX{{\rm B\kern-.05em{\sc i\kern-.025em b}\kern-.08em
    T\kern-.1667em\lower.7ex\hbox{E}\kern-.125emX}}

\begin{document}
\title{
\secrevised{Recursively Feasible} Chance-constrained Model Predictive Control under Gaussian Mixture \\ Model Uncertainty}
\author{Kai Ren, Colin Chen, Hyeontae Sung, Heejin Ahn, Ian Mitchell and Maryam Kamgarpour	\thanks{Ren and Kamgarpour are with the SYCAMORE Lab, École Polytechnique Fédérale de Lausanne (EPFL), Switzerland (e-mail: {\tt\small kai.ren@epfl.ch; maryam.kamgarpour@epfl.ch}). Chen and Mitchell are with the Department of Computer Science, The University of British Columbia, Vancouver, BC, Canada (e-mail: {\tt\small mitchell@cs.ubc.ca; colinc@fastmail.com}). Sung and Ahn are with the School of Electrical Engineering, KAIST, South Korea (e-mail: {\tt\small hyeontae.sung@kaist.ac.kr; heejin.ahn@kaist.ac.kr}).}
\thanks{
Ren and Kamgarpour's research is supposed by Swiss National Foundation Grant $\#200020\_207984 \slash  1$. Mitchell, Chen, and Ahn were partially supported by National Science and Engineering Research Council of Canada (NSERC) Discovery Grant \#RGPIN-2017-04543. Sung and Ahn were partially supported by the Ministry of Science and ICT, Korea, under the Information Technology Research Center support program (IITP-2024-RS-2023-00259991) supervised by the Institute for Information \& Communications Technology Planning \& Evaluation.}
}

\maketitle
\begin{abstract}
We present a chance-constrained model predictive control (MPC) framework under Gaussian mixture model (GMM) uncertainty. Specifically, we consider the uncertainty that arises from predicting future behaviors of moving obstacles, which may exhibit multiple modes (for example, turning left or right). To address the multi-modal uncertainty distribution, we propose three MPC formulations: nominal chance-constrained planning, robust chance-constrained planning, and contingency planning. We prove that closed-loop trajectories generated by the three planners are safe. The approaches differ in conservativeness and performance guarantee. In particular, the robust chance-constrained planner is recursively feasible under certain assumptions on the propagation of prediction uncertainty. On the other hand, the contingency planner generates a less conservative closed-loop trajectory than the nominal planner. We validate our planners using state-of-the-art trajectory prediction algorithms in autonomous driving simulators. 
\end{abstract}

\begin{IEEEkeywords}
Autonomous vehicles; Stochastic optimal control; Trajectory planning.
\end{IEEEkeywords}

\section{Introduction} \label{sec:introduction}

\IEEEPARstart{A}{utonomous} systems, such as self-driving cars, face significant challenges when operating in dynamic environments. For instance, an autonomous vehicle must plan its trajectory while avoiding collisions with other road agents, such as other vehicles and pedestrians. The future movements of other road agents are often multi-modal, for example, a nearby vehicle may go straight or turn at intersections. Planning safe trajectories in scenarios involving multi-modal behavior of road agents (e.g., intersection, overtaking) is particularly challenging yet essential for real-world applications. Our work focuses on developing a provably safe and computationally tractable trajectory planning approach in the presence of multi-modal uncertainty.


\subsection{Related Works}

Risk-constrained trajectory planning is a common formulation to ensure safety in uncertain environments. This formulation, rather than enforcing constraints for all uncertainty realizations, tolerates constraint violation up to a given threshold in terms of a chosen risk metric. The chance constraint is a common risk metric that restricts the probability of constraint violation. However, the chance constraint formulation is unaware of the potential severity of constraint violation. This has motivated the use of conditional value-at-risk (CVaR) \cite{Majumdar2017, Hakobyan2019, Barbosa2021}, which accounts for the expected amount of constraint violation. Both chance-constrained and CVaR-constrained trajectory planning problems are tractable when constraints are linear, and the uncertainty follows a (uni-modal) Gaussian distribution \cite{ Calafiore2006CC, Blackmore2011,jha-jar18, ahmadi2011}. 

In the context of autonomous driving, the trajectory prediction models \cite{Salzmann2020, Rhinehart2019} show that in complex and interactive environments, the probability distributions over the future positions of the vehicles are multi-modal. These past works have used a Gaussian mixture model (GMM) to depict the multi-modal uncertainty. Motivated by the prevalence of GMM in trajectory predictions, our work focuses on developing a framework for safe trajectory planning under GMM uncertainty. When the moments of each mode of GMM uncertainty are known and constraints are linear, chance-constrained problems can be tractably solved \cite{Hu2022}, while CVaR-constrained problems can be addressed by an iterative cutting-plane \cite{You2021GMMCVaR} approach.

However, the exact distribution may be unknown in practice, and only samples (e.g., sensor data) may be available. Without any prior knowledge of the uncertainty's distribution, CVaR-constrained problems have been tractably addressed by a sample average approximation \cite{Hakobyan2019} without any theoretical approximation guarantees. On the other hand, a distributionally robust approach with Wasserstein distance ambiguity set \cite{Esfahani2015DRO} can be applied to trajectory planning \cite{Hakobyan2020}. Under certain assumptions on the risk constraints, a finite-sample safety guarantee has been derived for this approach. With partial knowledge of the uncertainty's distribution, e.g., the uncertainty is multi-modal and the number of modes is known, \cite{Ahn_2022} extends the scenario approach \cite{Calafiore2006Scenario, campi2009} to address trajectory planning under multi-modal uncertainty. On the other hand, \cite{Ren2023} estimates the GMM moments by clustering the samples and robustifies the risk constraints against estimation error using moment concentration bounds. However, in the context of trajectory planning, both our prior studies \cite{Ahn_2022, Ren2023} are based on an open-loop framework.

In autonomous vehicle trajectory planning, overly cautious behavior or infeasibility issues are frequently observed in fast-changing and highly uncertain environments. Model Predictive Control (MPC) \cite{Garcia1989MPC, schwarm1999, Mesbah2016} address the problem by solving the trajectory planning problem in a receding-horizon manner. Leveraging the advancements in highly accurate sensors and real-time forecasting models, one can plan the ego vehicle's (EV's) trajectory with frequently updated observations and predictions in a closed-loop fashion. Past works on autonomous driving developed MPC under GMM uncertainty \cite{Nair2022}, showcasing its efficacy through simulations. However, no theoretical guarantees were provided.

One major obstacle hindering the safety of MPC is the recursive feasibility problem. It requires the existence of a feasible control input for each time step; the absence of such a solution indicates a potential violation of the safety constraints. This property has been well-studied in MPC under deterministic settings \cite{Kerrigan2000}. Among stochastic MPC frameworks, some previous works \cite{Schildbach2014} make direct recursive feasibility assumptions. In contrast, when the uncertainty has bounded support, other research focuses on proving recursive feasibility \cite{Kouvaritakis2021} or identifying safe control invariant sets \cite{Kuwata2009} where the system can deterministically remain feasible for an indefinite period of time. However, defining such sets becomes difficult for risk-constrained trajectory planning with unbounded GMM uncertainties.

\subsection{Contributions}
We develop an MPC framework for autonomous driving under GMM uncertainty with a provable safety guarantee. Our main contributions are:
\begin{itemize}
    \item Developing a recursively feasible and safe MPC framework. To this end, we develop conditions on the propagation of the GMM prediction over planning steps, and tighten the constraints of the MPC accordingly to robustify against the prediction propagation;

    \item Formulating an alternative planning scheme based on contingency planning \cite{hardy2013},\revised{\cite{Alsterda2019contMPC},} for reducing conservatism of chance-constrained MPC under GMM uncertainty while preserving the safety guarantees;

    \item Validating our methods by simulating with the state-of-the-art trajectory prediction model \textit{Trajectron++} \cite{Salzmann2020} in an autonomous driving simulator CARLA \cite{Dosovitskiy17}. 
\end{itemize}

The rest of the paper is organized as follows. In Section \ref{sec:CLTP}, we formulate the chance-constrained MPC problems with GMM uncertain parameters and present a planning framework combining shrinking and receding horizon planning schemes. Section \ref{section:unifiedPlanning} presents an MPC scheme and develops the recursive feasibility and safety guarantees under assumptions on the GMM predictions. Furthermore, we present an alternative contingency planning scheme for reducing conservatism without losing the safety guarantee. Section \ref{sec:experiment} demonstrates our methods on a state-of-the-art autonomous driving simulator.

\textit{Notation:} A Gaussian distribution with mean $\mu$ and covariance matrix $\Sigma$ is denoted as $\mathcal{N}(\mu, \Sigma)$. By $\text{\textPsi}^{-1}(\cdot)$, we denote the inverse cumulative distribution function of the standard Gaussian distribution $\mathcal{N}(0, 1)$. We denote a set of consecutive integers $\{a, a+1, \dots,b\}$ by $\mathbb{Z}_{a:b}$. We denote the conjunction by $\bigwedge$ and the disjunction by $\bigvee$. The weighted \(L_2\) norm of a vector \( x \) with a weighting matrix \( P \succ 0\) is $\|x \|_{P} = \sqrt{x^T P x}$.

\section{Problem Formulation} \label{sec:CLTP}

Our goal is to find the optimal ego vehicle (EV) state trajectory that minimizes a specific cost function (e.g., fuel consumption or distance to the target) while ensuring a collision-free path in the presence of other vehicles (OVs) throughout the planning horizon. The EV is modeled as a deterministic linear time-varying system
\begin{equation} \label{eq:systemdynamic}
    x_{t+1} = A_t x_t + B_t u_t,
\end{equation}
where $x_t \in \mathbb{R}^{n_x} $ and $u_t \in \mathbb{R}^{n_u}$ denote the state and input at time $t$, and $A_t \in \mathbb{R}^{n_x \times n_x}$ and $B_t \in \mathbb{R}^{n_x \times n_u}$ are system's dynamics matrices. To comply with road rules such as staying within lanes and adhering to speed limits, the state and control input are confined within time-invariant convex sets,
\begin{equation} \label{eq:xuconstraints}
x_t \in \mathcal{X} \text{ and } u_t \in \mathcal{U}.
\end{equation}

Given an initial state $x_0$ and a planning horizon $T$, the trajectory planning framework generates a sequence of control inputs $\textbf{u} = (u_0, u_1, \dots, u_{T-1})$. Given the initial state $x_0$ and the system dynamics \eqref{eq:systemdynamic}, the input sequence $\textbf{u}$ would lead to a state trajectory $\textbf{x} = (x_0, x_1, \dots, x_{T})$ that should be free from collisions with $J$  other vehicles (OVs) while satisfying the constraints \eqref{eq:xuconstraints}. The $j$-th OV is modeled as a polytope with $I_j$ edges. For brevity, we use $\mathcal{X}_{t,j}^{safe}$ to denote the set of the EV states that lie outside of the polytope corresponding to the $j^{th}$ OV at time $t$. This collision avoidance set can be represented as the EV states that lie beyond at least one edge of the OV polytope. Hence, $\mathcal{X}_{t,j}^{safe}$ can be written as a disjunction of linear constraints: $x_{t} \in \mathcal{X}_{t,j}^{safe} \Leftrightarrow \bigvee_{i=1}^{I_j} \delta_{ij}^{t^\intercal} \Tilde{x}_t \leq 0$, where $\Tilde{x}_t = [x_t^{\intercal} \;\; 1 ]^{\intercal}$ and the uncertainty $\delta_{ij}^{t} \in \mathbb{R}^{n_x+1}$ corresponds to the $i^{th}$ face of the $j^{th}$ OV at time $t$. A 2-dimensional space example, i.e., $n_x = 2$, is shown in Fig.~\ref{fig:linearconstraint}, where the EV states lying in the shaded orange region are considered to be free from collision with the OV. 

\begin{figure}[h]
    \centering \vspace{-1mm}
    \includegraphics[width=6.5cm]{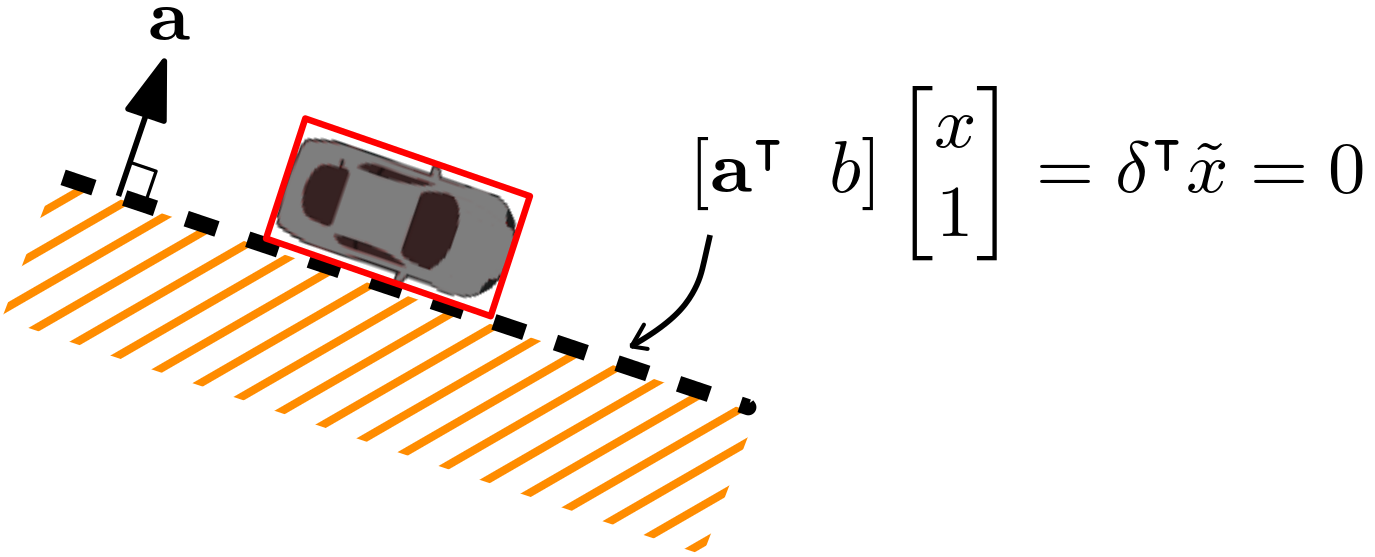}
    \caption{Polyhedral OV and linear constraint for collision avoidance in 2-dimensional space.}
    \label{fig:linearconstraint}
\end{figure}

To capture the uncertain and multi-modal behavior of the OVs  (e.g., an OV may go straight or turn at an intersection), we use a GMM to model the underlying distribution of the OV uncertainty $\delta_{ij}^{t}$ and enforce the following assumption.
\begin{assumption} \label{assumption:guassianConstraint}
    \textbf{[GMM uncertainty]}
    The multi-modal uncertainty $\delta_{ij}^{t}$ is GMM distributed with $K_j$ modes of behavior, i.e., $\delta_{ij}^{t} \sim \sum_{k=1}^{K_j} \pi_k \cdot \cN (\mu_{ijk}^{t}, \Sigma_{ijk}^{t})$.
\end{assumption}

\subsection{Selection of Risk-constrained Formulation} \label{subsection:riskSelection}

To address the GMM uncertainty, we consider a risk-constrained formulation to ensure safety. To choose an appropriate risk metric and approach, we investigate various different risk-constrained methods. \revised{For the risk metric, we consider a chance constraint or a CVaR constraint. For each of the two risk formulations, we apply existing tractable approaches based on the assumptions on the uncertainty's distribution. In particular, we explore cases where the distribution is fully known, partially known, or completely unknown. We evaluate the performances of the risk-constrained approaches based on a simple yet instructive example, consisting of a single risk-constrained optimization problem with a linear constraint function. The detailed descriptions of the optimization problem, the risk-constrained formulations and the comparative analysis on their performances are detailed in Appendix~\ref{Appendix:risk-constrained}}.

The key insight from this simple numerical study is that for the multi-modal uncertainty, the methods which employ chance constraint as a risk metric and incorporate distribution information, tend to yield less conservative solutions while ensuring computational efficiency. In contrast, enforcing CVaR constraints or distributional robustness may lead to an excessively conservative solution and significantly extend computation time. The safety constraint in autonomous car trajectory planning problems can be defined as a conjunction of single risk constraints for tractability \cite{Lefkopoulos2021, Hakobyan2019, dixit2021}, and we expect that the performances of the risk-constrained methods discussed in this subsection for a single risk constraint can be extended to such trajectory planning problems. To ensure that our trajectory planning framework is not overly conservative and to consider real-time planning for autonomous driving, we focus on the chance-constrained formulation in this work. Furthermore, our work focuses on ensuring recursive feasibility in MPC. \secrevised{Developing recursive feasibility conditions with data-driven methods, such as scenario approach \ref{num:scenario}, is challenging without a parameterized uncertainty distribution.} Therefore, we investigate the propagation of the uncertainty's distribution model (i.e., GMM) over planning steps and adopt CC-MTA \ref{num:MTA} and CC-MRA \ref{num:MRA} to handle the chance constraints. 

\subsection{Chance-constrained Planning} \label{subsection:shrinkingReceding}
 
Recall the definition of the safe set of EV states, denoted as $\mathcal{X}_{t,j}^{safe}$; it represents the set of EV states that are free from collision with the $j^{th}$ OV at time $t$. To capture the collision avoidance in trajectory planning with $J$ OVs over a planning horizon $T$, we formulate a \textit{joint} chance constraint
\begin{equation} \label{eq:MPCchance}
    \mathbb{P}\left(\bigwedge_{t=1}^{T} \bigwedge_{j=1}^{J} x_{t} \in \mathcal{X}_{t,j}^{safe} \right) \geq 1-\epsilon.
\end{equation}

 The joint chance constraints can be transformed into tractable deterministic constraints. In particular, we first use the Big-M method \cite{Schouwenaars2001} to handle the non-convexity of the disjunctive collision avoidance constraint (as illustrated in Fig.~\ref{fig:linearconstraint}): $x_{t} \in \mathcal{X}_{t,j}^{safe} \Leftrightarrow \bigvee_{i=1}^{I_j} \delta_{ij}^{t^\intercal} \Tilde{x}_t \leq 0 \Leftrightarrow \bigwedge_{i=1}^{I_j} \delta_{ij}^{t^\intercal} \Tilde{x}_t \leq Mz_{ij}^{t}$, where $M$ is a positive number that is larger than any possible value of the linear constraint and $z_{ij}^{t} \in \{0,1\}$ is a binary variable. Then, we use Boole's inequality \cite{CaseBerg:01} to decouple the joint chance constraint into multiple single chance constraints. Based on Assumption~\ref{assumption:guassianConstraint}, $\delta_{ij}^{t}$ conforms to a GMM. Given the moments $\mu_{ijk}^t$ and $\Sigma_{ijk}^t$, either known or estimated from samples, the single chance constraints can then be equivalently reformulated into second-order conic constraints. The resulting constraints are as follows.
\begin{subequations} \label{eq:chanceGMM}
    \begin{alignat} {2}
    & \bigwedge_{t=1}^{T} \bigwedge_{j=1}^{J} \bigwedge_{i=1}^{I_j} \bigwedge_{k=1}^{K_j} \Gamma_{ijk}^{t} \sqrt{ \Tilde{x}_t^{\intercal}{\Sigma_{ijk}^{t}} \Tilde{x}_t} + \Tilde{x}_t^{\intercal}\mu_{ijk}^{t} \leq Mz_{ijk}^t, \label{eq:OLconjunctionCC} \\
     & \sum_{k=1}^{K_j} \pi_k \epsilon_{ijk}^{t} = \epsilon_{ij}^{t} = \epsilon/(TJ). \label{constraint:RAtj} \\
    & \sum_{i=1}^{I_j} z_{ijk}^{t} = I_j - 1, \;\;\; z_{ijk}^{t} \in \{0, 1\},\label{constraint:RAandBinary}
    \end{alignat}
\end{subequations}
where $\Gamma_{ijk}^{t} := \text{\textPsi}^{-1}(1-\epsilon_{ijk}^{t})$. Note that the reformulation is conservative (i.e., $\eqref{eq:chanceGMM} \Rightarrow \eqref{eq:MPCchance}$) due to the uniform allocation of $\epsilon$ to all OVs and timesteps \cite[Lemma 1]{Lefkopoulos2021}. 

Previous work \cite{Ahn_2022, Ren2023} used an open-loop trajectory planning approach, which only supports a short planning horizon. In particular, when the planning horizon $T$ increases, the OV prediction becomes highly uncertain such that no feasible solution exists. For instance, our previous experimental results showed that using the \textit{Trajectron++} model \cite{Salzmann2020} on real-world intersection data to forecast a vehicle's position 4 seconds ahead results in a substantial state variation of up to 28 meters, posing a significant challenge to planning feasible trajectories for an EV. To address the infeasibility issues of open-loop planning over a long planning horizon, \cite{Nair2022} proposed an MPC framework under GMM uncertainty. However, the issues of recursive feasibility and safety guarantees were not addressed. 


To this end, our goal is to develop an MPC framework under GMM uncertainty with provable recursive feasibility and safety guarantees. In the rest of the paper, we use $\tau$ to denote the MPC planning steps and $T_{OH}$ to denote the open-loop planning horizon at each planning step. At each $\tau \in \{0, 1, \ldots, T-1\}$, we update the GMM moment predictions: $(\mu_{ijk}^{t|\tau}, \Sigma_{ijk}^{t|\tau})$ for future time steps $t = \tau+1, \tau+2, \ldots, T_{OH}$ based on the prediction module or sample estimation. Utilizing the predictions, we iteratively solve an open-loop trajectory planning problem up to $T_{OH}$. 

We consider two common MPC planning horizon settings, as shown in Figure \ref{fig:CLstatemachine}.

\begin{itemize}
    \item \textbf{Shrinking-horizon:} $T_{OH} = T$, where we compute trajectories up to the end of the closed-loop planning horizon at each planning step $\tau$. 

    \item \textbf{Receding-horizon:} $T_{OH} = \tau + T_s$, where we plan a trajectory for the upcoming $T_s$ timesteps ($T_s < T$) at each planning step. 
\end{itemize}

\begin{figure}[!t]
     \centering
     \includegraphics[width=7cm]{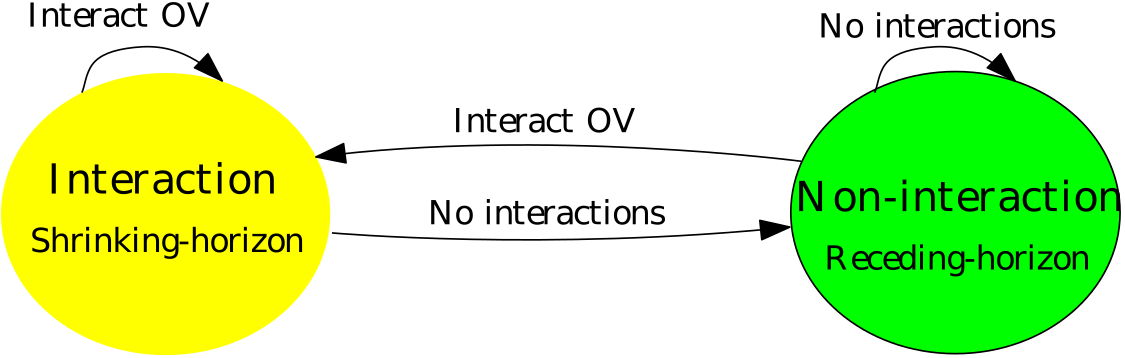}
     \caption{High-level state machine for MPC combining the shrinking-horizon and receding-horizon settings.}
     \label{fig:CLstatemachine}
 \end{figure}

The shrinking-horizon setting is employed when the EV needs to complete a maneuver, such as a turn or a lane merge, within a finite time window and in the presence of OVs. 
Receding-horizon planning allows a longer closed-loop planning horizon and is used in non-interactive scenarios, such as lane keeping and car following. The transition from a non-interaction state to an interaction state occurs when the EV has a fixed-horizon goal in the presence of OVs.
The occurrence of the state transition can be determined by the sensor observations.
Since the OV-involved scenarios present the most challenge regarding uncertainties in future trajectories of OVs, our focus will be on deriving an MPC algorithm with provable recursive feasibility and safety in the shrinking-horizon setting.

\section{Chance-constrained MPC} \label{section:unifiedPlanning}
 At every MPC planning step $\tau = 0, 1, \ldots, T-1$, we solve a trajectory planning problem $\cP$, to be specified in the following sections. The solution to the planning problem is the control input $\textbf{u}_{\tau} = (u_{\tau|\tau}, u_{\tau+1|\tau}, \dots, u_{T_{OH}-1|\tau})$ and open-loop state trajectory $\textbf{x}_{\tau} = (x_{\tau+1|\tau}, x_{\tau+2|\tau}, \dots, x_{T_{OH}|\tau})$. We actuate the EV with the first control input $u_{\tau|\tau}$ and proceed to the next time step. The algorithm that implements MPC over a closed-loop planning horizon $T$ is outlined in the pseudocode below.
\begin{algorithm}[H]
\begin{minipage}{0.5\textwidth}
\caption{Chance-constrained MPC}
\label{alg:MPCplanning}
\begin{algorithmic}[1]
\State \textbf{Initialize:} initial state $x_0$, $\tau=0$, $T_{OH}$,  problem $\cP$

\For{$\tau = 0, 1, 2, \ldots, T-1$}
\State \textbf{Obtain} GMM prediction: $(\mu^{t|\tau}, \Sigma^{t|\tau}),\; \forall t \in \mathbb{Z}_{\tau+1:T_{OH}}$
\State \textbf{Solve} $\cP$ \label{alg:open-loop}
\If{feasible}
\State \textbf{Obtain: } $\textbf{u}_\tau$, $\textbf{x}_\tau$
\State Actuate EV with $u_{\tau|\tau}$, reach $x_{\tau+1|\tau}$
\Else
\State \Return Infeasible \label{algState:infeasible}
\EndIf
\EndFor
\State \Return $\textbf{x} = \left(x_{(1|0)}, x_{(2|1)}, x_{(3|2)}, \ldots, x_{(T|T-1)} \right)$
\end{algorithmic}
\end{minipage}
\end{algorithm} 

In the following, we present three approaches to formulating $\cP$ in Algorithm~\ref{alg:MPCplanning}. In subsection \ref{subsection:nominalPlanning}, we present a nominal planner, which is a naive deterministic reformulation of chance-constrained MPC under GMM. In subsection \ref{sec:robustPlanning}, we develop a robust MPC planner that implements tightened chance constraints and is proved to be recursively feasible under certain assumptions on the propagation of the GMM predictions across planning steps. Inspired by \cite{hardy2013}, in subsection \ref{section:contingencyPlanning}, we apply the concept of contingency planning to generate less conservative trajectory than the nominal and robust planners.

\subsection{Nominal Planning} \label{subsection:nominalPlanning}
Our first approach to formulating $\cP$ in Algorithm \ref{alg:MPCplanning} is based on the deterministic formulation of the chance constraints as introduced in \eqref{eq:OLconjunctionCC}.
\begin{subequations} \label{problem:CLCCTPproblem}
    \begin{alignat} {2}
        & \underset{\textbf{u}_{\tau}, \; \textbf{z}_{\tau}}{\text{min}}
        & \quad & \cJ(\textbf{u}_{\tau}, \textbf{x}_{\tau}) \\
        & \;\;\; \text{s.t.}
        & & x_{t+1} = A_t x_t + B_t u_t,  \;\;\; \forall t \in \mathbb{Z}_{\tau: (T_{OH}-1)}, \label{CL:dynamics}\\
        & & & \Gamma^{t} \sqrt{ \Tilde{x}_{t|\tau}^{\intercal}{\Sigma^{t|\tau}} \Tilde{x}_{t|\tau}} + \Tilde{x}_{t|\tau}^{\intercal}\mu^{t|\tau} \leq Mz^{t|\tau}, \notag\\ 
        & & & \;\;\;\;\;\;\;\;\;\;\;\;\;\;\;\;\;\;\;\;\;\;\;\;\;\;\;\;\;\;\;\;\;\;\;\;\;\;\;\;  \forall t\in \mathbb{Z}_{(\tau+1):T_{OH}}, \label{CL:constraint} \\
        &&& \eqref{eq:xuconstraints}, \eqref{constraint:RAtj}, \eqref{constraint:RAandBinary}. \notag 
\end{alignat}
\end{subequations}
Note that in \eqref{CL:constraint} and for the rest of this section, for brevity, we have excluded the sub-indices corresponding to OV $j$, mode number $k$, and the OV faces $i$ (see \eqref{eq:OLconjunctionCC}). In particular, constraint \eqref{CL:constraint} must be satisfied for all $j \in \mathbb{Z}_{1:J}, \; k \in \mathbb{Z}_{1:K_j}, \; i \in \mathbb{Z}_{1:I_j}$. 

\revised{
\begin{remark} \label{remark:RHfeasibility} \textbf{[Receding-horizon recursive feasibility]}
    A major challenge in MPC is to ensure recursive feasibility (i.e., Algorithm \ref{alg:MPCplanning} does not go to line \ref{algState:infeasible}). In the receding-horizon state, we consider non-interactive scenarios. First, let us consider cases when the constraints related to the OVs are nonexistent, that is, \eqref{constraint:RAtj}, \eqref{constraint:RAandBinary}, \eqref{CL:constraint} do not exist. In such cases, Algorithm \ref{alg:MPCplanning} is recursively feasible when the terminal state $x_T$ lies in a terminal set that is control invariant for the EV dynamics \eqref{eq:systemdynamic} subject to the constraint \eqref{eq:xuconstraints} \cite[Thereom 12.1]{Borrelli2017}. In an alternative cases when the uncertainty admits a bounded support,
    the recursive feasibility can be ensured by using the robust constraint tightening against the worst-case uncertainty realization \cite{Lorenzen2017, Kouvaritakis2021}. In our MPC framework, we employ receding-horizon MPC when there exists no OV-related constraint. We assume that a control invariant terminal state set exists, and thus, the receding-horizon MPC is recursively feasible.   
\end{remark}
}

In autonomous driving scenarios where \secrevised{the EV needs to complete a maneuver within a finite time window and }the OV states are uncertain, we employ shrinking-horizon planning. In the following subsection, we explore conditions on the growth of the GMM uncertainty along the planning steps and robustification of chance constraints against the worst-case uncertainty propagation. Under these conditions, we can ensure the recursive feasibility of shrinking-horizon MPC.

\subsection{Robust Planning for Recursive Feasibility} \label{sec:robustPlanning}
\secrevised{
In this subsection, we explore the conditions for GMM uncertainty propagation and tighten the chance constraints \eqref{CL:constraint} to ensure recursive feasibility. The idea is an intuitive driving behavior during interactions with the OVs. At the beginning of an interaction, we consider the worst-case behaviors of the OVs. As time progresses, the future states and driving patterns of the OVs become clearer, making it easier for the EV to remain feasible and decide how to drive.}

During interactions between the EV and the OVs, our empirical findings from the forecasting model show the rather intuitive fact that prediction uncertainty reduces over time with updated observations of an OV. Moreover, the OV's multi-modal behavior, such as the possibility of going straight or turning at intersections, gets reduced to one mode with ongoing observation over time. These observations motivate the following assumptions about the shrinkage of the OV prediction uncertainty in interactive scenarios. These assumptions are, in turn, used to ensure the recursive feasibility of a shrinking-horizon MPC problem.  
\begin{assumption} \label{assumption:reducingModeNum}
    The number of the prediction modes does not increase with planning steps, that is, $K^{\tau} \geq   K^{\tau+1}$.
\end{assumption}

To quantify the difference in the prediction at time $\tau$ and $\tau +1$ of the OV position at time $t$, we investigate the propagation of the predicted GMM moments corresponding to each face of each OV under each mode of behavior. Let us define the Frobenius norm of the uncertainty's covariance matrix between two consecutive time steps as
\begin{equation} \label{eq:covShrink}
    \sqrt{\norm{\Sigma^{t|\tau}}_F} - \sqrt{\norm{\Sigma^{t|\tau+1}}_F} = g^t(\tau).
\end{equation}
The shift of the mean prediction between two consecutive time steps is denoted by
\begin{equation} \label{eq:meanShift}
        \norm{\mu^{t|\tau} - \mu^{t|\tau+1}}_2 = h^t(\tau).
\end{equation}
\begin{assumption}\label{assumption:VaRshrink} 
For any  $t \in \mathbb{Z}_{2:T}$ and planning steps $\tau \in \mathbb{Z}_{0:t-1}$, the shift of the mean prediction is bounded by the covariance shift as follows: \label{subass:covBoundMean}
    \begin{equation} \label{eq:covBoundmean}
        h^t(\tau) \leq \Gamma^{t} \cdot g^t(\tau),
    \end{equation}
where $\Gamma^t=\text{\textPsi}^{-1}(1-\epsilon^{t})$ is the $(1-\epsilon^t)$ quantile of the standard Gaussian distribution. 
\end{assumption}

To ensure recursive feasibility, we introduce the following robust planning scheme $\cP$ in Algorithm \ref{alg:MPCplanning}.
\begin{subequations} \label{problem:RobustMPC}
    \begin{alignat} {2}
        & \underset{\textbf{u}_{\tau}, \; \textbf{z}_{\tau}}{\text{min}}
        & \quad & \cJ(\textbf{u}_{\tau}, \textbf{x}_{\tau}) \\
        & \;\;\; \text{s.t.}
        & & \Gamma^{t}  \sqrt{\norm{\Sigma^{t|\tau}}_F} \cdot \norm{\Tilde{x}_{t|\tau}}_2 + \Tilde{x}_{t|\tau}^{\intercal}  {\mu^{t|\tau}} \leq M z^{t|\tau}, \notag\\
        &&& \;\;\;\;\;\;\;\;\;\;\;\;\;\;\;\;\;\;\;\;\;\;\;\;\;\;\;\;\;\;\;\;\;\;\;\;\;\;\; \;\;\;\;\;\forall t\in \mathbb{Z}_{(\tau+1):T_{OH}}, \label{eq:worstVaRconstraint}\\
        &&& \eqref{eq:xuconstraints}, \eqref{CL:dynamics}, \eqref{constraint:RAtj}, \eqref{constraint:RAandBinary}.
\end{alignat}
\end{subequations}
Note that a feasible solution of \eqref{problem:RobustMPC} is also a feasible solution of \eqref{problem:CLCCTPproblem}, because \eqref{eq:worstVaRconstraint} is a conservative approximation of \eqref{CL:constraint} based on Cauchy–Schwarz inequality,
\begin{equation} \label{eq:Cauchy}
    \sqrt{ \Tilde{x}_{t|\tau}^{\intercal}{\Sigma^{t|\tau}} \Tilde{x}_{t|\tau}} \leq \sqrt{\norm{\Sigma^{t|\tau}}_F} \cdot \norm{\Tilde{x}_{t|\tau}}_2.
\end{equation}


\begin{proposition}\textbf{[Shrinking-horizon recursive feasibility]} \label{theorem:recursivefeasibility}
Under Assumptions~\ref{assumption:guassianConstraint}--\ref{assumption:VaRshrink}, the existence of a feasible solution for $\eqref{problem:RobustMPC}$ at $\tau=0$ ensures that Algorithm \ref{alg:MPCplanning} with $\cP$ set to \eqref{problem:RobustMPC} is recursively feasible.
\end{proposition}

\begin{proof}
Suppose at some MPC planning step $\tau$, there exists a feasible solution of \eqref{problem:RobustMPC}, where the state trajectory is $\textbf{x}_\tau = (x_{\tau+1|\tau}, \dots, x_{T|\tau})$. At the next MPC planning step $\tau+1$, for any $t\in \mathbb{Z}_{(\tau+2):T},$

\begin{align}
\begin{split}\label{proof:unifiedRF}
    & \; \Gamma^{t}  \sqrt{\norm{\Sigma^{t|\tau+1}}_F} \cdot \norm{\Tilde{x}_{t|\tau+1}}_2  + \Tilde{x}_{t|\tau+1}^{\intercal}{\mu^{t|\tau+1}} \\
     \leq & \; \Gamma^{t}  \sqrt{\norm{\Sigma^{t|\tau}}_F} \cdot \norm{\Tilde{x}_{t|\tau+1}}_2 + \Tilde{x}_{t|\tau+1}^{\intercal}{\mu^{t|\tau}} \\ 
     & \;\;\;\; + \biggl( h^t(\tau) - \Gamma^{t} \cdot g^t(\tau)\biggr) \cdot \norm{\Tilde{x}_{t|\tau+1}}_2 \\
      \leq & \; \Gamma^{t}  \sqrt{\norm{\Sigma^{t|\tau}}_F} \cdot \norm{\Tilde{x}_{t|\tau+1}}_2 + \Tilde{x}_{t|\tau+1}^{\intercal}{\mu^{t|\tau}}.
\end{split}
\end{align}
The first inequality is due to \eqref{eq:covShrink} and because from Cauchy-Schwarz inequality,
\begin{align*}
    \Tilde{x}_{t|\tau+1}^{\intercal} \left( {\mu^{t|\tau+1}} -{\mu^{t|\tau}}\right) &\leq \norm{\Tilde{x}_{t|\tau+1}}_2 \cdot \norm{\mu^{t|\tau+1} -{\mu^{t|\tau}}}_2
\end{align*} 
and therefore, by \eqref{eq:meanShift}, 
\begin{align*}
    \Tilde{x}_{t|\tau+1}^{\intercal} {\mu^{t|\tau+1}} \leq  \Tilde{x}_{t|\tau+1}^{\intercal} {\mu^{t|\tau}} +\norm{\Tilde{x}_{t|\tau+1}}_2 h^t(\tau).
\end{align*} 
The second inequality in \eqref{proof:unifiedRF} follows from \eqref{eq:covBoundmean}. If we set a state trajectory at $\tau+1$ to be $(x_{\tau+2|\tau+1}, \ldots, x_{T|\tau+1}) = (x_{\tau+2|\tau}, \ldots, x_{T|\tau})$, then it satisfies 
$$\Gamma^{t}  \sqrt{\norm{\Sigma^{t|\tau+1}}_F} \cdot \norm{\Tilde{x}_{t|\tau+1}}_2 + \Tilde{x}_{t|\tau+1}^{\intercal}{\mu^{t|\tau+1}} \leq M z^{t|\tau}$$
by \eqref{proof:unifiedRF}. Note that $\Gamma^t$ remains the same throughout the proof. This is because $\Gamma^t$ depends only on $\epsilon^t$, which is the same for each mode $k$ and does not vary at different planning steps due to the uniform risk allocation \eqref{constraint:RAtj}. Therefore, Algorithm \ref{alg:MPCplanning} with $\cP$ set to \eqref{problem:RobustMPC} is recursively feasible.
\end{proof}


\revised{
The high-level planning state machine in Fig.~\ref{fig:CLstatemachine} comprises two states and two transitions between the states. We have now discussed the recursive feasibility conditions for MPC in non-interaction (receding-horizon) and interaction (shrinking-horizon) states. To ensure long-term maneuverability with frequent state transitions, it's essential to examine feasibility conditions for the transitions from one state to another.
\begin{remark} \textbf{[State Transition Feasibility]}
    When transitioning from interaction to non-interaction state, the existence of a control-invariant terminal set (as discussed in Remark~\ref{remark:RHfeasibility}) ensures the feasibility. To ensure the transition from receding to shrinking horizon planning remains feasible, we need to ensure that a feasible solution can be found for \eqref{problem:RobustMPC} at $\tau=0$ when the EV has a potential interaction with the OVs. Otherwise, the EV should not transition to the interaction state. For instance, if the EV attempts a lane change but encounters multiple OVs in the neighboring lane, preventing the completion of the maneuver, the EV must continue driving in its original lane.
\end{remark}
}

Our next result shows that the planned trajectory satisfies the original chance constraint \eqref{eq:MPCchance} over the closed-loop horizon $T$.
\begin{theorem}\textbf{[Safety guarantee of MPC]} \label{theorem:MPCfeasibility}
\begin{enumerate}[leftmargin=14pt, label=\arabic*)]
    \item Under Assumption~\ref{assumption:guassianConstraint}, if Algorithm \ref{alg:MPCplanning} with $\cP$ set to \eqref{problem:CLCCTPproblem} does not return Infeasible and generates a closed-loop state trajectory $\textnormal{\textbf{x}}$, then $\textnormal{\textbf{x}}$ satisfies the original chance constraint \eqref{eq:MPCchance}.
    \item Under Assumptions~\ref{assumption:guassianConstraint}--\ref{assumption:VaRshrink}, Algorithm \ref{alg:MPCplanning} with $\cP$ set to \eqref{problem:RobustMPC} never returns Infeasible and generates a closed-loop trajectory $\textnormal{\textbf{x}}$ if \eqref{problem:RobustMPC} is feasible with the initial state $x_0$. Moreover, $\textnormal{\textbf{x}}$ satisfies the original chance constraint \eqref{eq:MPCchance}.
\end{enumerate} 
\end{theorem}

\begin{proof}
1) If Algorithm \ref{alg:MPCplanning} with $\cP$ set to \eqref{problem:CLCCTPproblem} generates a closed-loop state trajectory $\textbf{x}=(x_{(1|0)}, \ldots, x_{(T|T-1)})$, then $x_{t|t-1}$ satisfies \eqref{CL:constraint} for all $t \in \mathbb{Z}_{1:T}$. Considering every face across all modes of all OVs: $j \in \mathbb{Z}_{1:J}, \; k \in \mathbb{Z}_{1:K_j}, \; i \in \mathbb{Z}_{1:I_j}$, the closed-loop trajectory $\textbf{x}$ also satisfies the conjunction of deterministic constraints \eqref{eq:OLconjunctionCC}. Given the implication of $\eqref{eq:chanceGMM}\Rightarrow\eqref{eq:MPCchance}$, we conclude that $\textbf{x}$ also satisfies \eqref{eq:MPCchance}.

\noindent
2) By Lemma~\ref{theorem:recursivefeasibility}, Algorithm~\ref{alg:MPCplanning} with $\cP$ set to \eqref{problem:RobustMPC} is recursively feasible and generates a closed-loop state trajectory $\textbf{x}=(x_{(1|0)}, \ldots, x_{(T|T-1)})$ that satisfies
$$\Gamma^{t}  \sqrt{\norm{\Sigma^{t|t-1}}_F} \cdot \norm{\Tilde{x}_{t|t-1}}_2 + \Tilde{x}_{t|t-1}^{\intercal} \cdot {\mu^{t|t-1}} \leq M z^{t|t-1}$$
 for all $t\in \mathbb{Z}_{1:T}$.  It follows from \eqref{eq:Cauchy} that $\textbf{x}$ also satisfies 
\begin{equation*}
 \Gamma^{t} \sqrt{ \Tilde{x}_{t|t-1}^{\intercal}{\Sigma^{t|t-1}} \Tilde{x}_{t|t-1}} + \Tilde{x}_{t|t-1}^{\intercal} \cdot \mu^{t|t-1} \leq Mz^{t|t-1}
\end{equation*}
for all $t\in \mathbb{Z}_{1:T}$. Considering every face across all modes of all OVs: $j \in \mathbb{Z}_{1:J}, \; k \in \mathbb{Z}_{1:K_j}, \; i \in \mathbb{Z}_{1:I_j}$, the closed-loop trajectory $\textbf{x}$ also satisfies \eqref{eq:OLconjunctionCC} and thus satisfies \eqref{eq:MPCchance}.
\end{proof}

\begin{remark} \textbf{[Probabilistic safety guarantee under estimated GMM moments]}
Theorem \ref{theorem:MPCfeasibility} provides a deterministic guarantee of the original chance constraint \eqref{eq:MPCchance}, assuming that the GMM moments are known. In many real-world applications, we only have samples of uncertainties rather than the exact knowledge of the moments \cite{Yoon2020}. In cases where the exact GMM moments are unknown and only $N_s$ samples of uncertainty are available, we can empirically cluster and estimate the GMM moments $(\hat{\mu}^{t}, \hat{\Sigma}^{t})$.

When the GMM moments are estimated from samples, we can robustify \eqref{CL:constraint} against moment estimation error \cite{Calafiore2006CC} using the following constraint. 
\begin{equation} \label{eq:MRA} \fontsize{9}{15}\selectfont
    \Gamma^{t}\sqrt{\left(1+r_{2}^t(N_s, \beta)\right)(\Tilde{x}_t^{\intercal}\hat{\Sigma}^{t}\Tilde{x}_t)} + r_{1}^t(N_s, \beta)\norm{\Tilde{x}_{t|\tau}}_2 + \Tilde{x}_t^{\intercal}\hat{\mu}^{t} \leq Mz^{t}.
\end{equation}

Here, $\beta$ is a user-defined probabilistic tolerance, and $r_1^t$, $r_2^t$ are the mean and covariance concentration bounds at time $t$, which take into account the difference between true and estimated moments. The robustification will provide a probabilistic safety guarantee. In particular, under Assumption~\ref{assumption:guassianConstraint}, if Algorithm \ref{alg:MPCplanning} with $\cP$ set to \eqref{problem:CLCCTPproblem} and \eqref{CL:constraint} being replaced by \eqref{eq:MRA} does not return infeasible. The closed-loop state trajectory $\textnormal{\textbf{x}}$ satisfies the original chance constraint \eqref{eq:MPCchance} with at least a $1-2TJ\beta$ probability \cite[Theorem 2]{Ren2023}.
\end{remark}


\subsection{Contingency Planning for Reducing Conservativeness}\label{section:contingencyPlanning}
A  challenge in trajectory planning under GMM is the high uncertainty of the predictions. If the EV must account for all modes of the OVs in a given planning horizon,
a feasible solution may not exist or may be too conservative. Here, we use a notion of contingency planning \cite{hardy2013},\revised{\cite{Alsterda2019contMPC, Chen2022brachMPC},} which accounts for predicted obstacle trajectory permutations, to reduce conservativeness under GMM uncertainty. In particular, we consider planning $L$ open-loop trajectories, each of which considers a subset of OV modes. We denote all mode indices of all OVs by $\cK := \bigcup_{j=1}^{J} \{1_j, 2_j, \ldots, K_j\}$. We select $L$ subsets of $\cK$: $(\cS_1, \ldots, \cS_L)$, where $ \bigcup_{l=1}^L \cS_l = \cK$. In this way, all OV modes will be considered at least once. For example, we have two OVs where OV$_1$ has two modes $\{1_1, 2_1\}$ (here the sub-index 1 represents mode belongs to OV$_1$), while OV$_2$ has one mode $\{1_2\}$. We can plan $L=2$ trajectories where trajectory $l=1$ considers modes $S_1 = \{1_1, 1_2\}$ and trajectory $l=2$ considers $S_2 = \{2_1 \}$.

For each trajectory $l$, we consider one subset of OV modes $\cS_l$, as shown in constraint \eqref{cst:contSafety}. All $L$ trajectories have the same control input and state for the first $T_c$ timesteps. The deterministic optimization problem solved at each MPC step $\tau$ is
\begin{subequations} \label{problem:contingencyMPC} 
    \begin{alignat} {2}
        \min_{\vu_\tau^{1:L}, \vz_\tau^{1:L}} \quad & \sum^L_{l=1} \cJ(\iter{\vu}{l}_\tau, \iter{\vx}{l}_\tau) \\
        \st\quad
        & \iter{x}{l}_{t + 1} = A_t \iter{x}{l}_t + B_t \iter{u}{l}_t, ~\forall t\in\bbZ_{\tau:T_{OH}-1}\label{cst:conincideDynamics}\\
        & \iter{u}{1}_{t|\tau} = \cdots = \iter{u}{L}_{t|\tau}, ~~~~~~~~\forall t\in\bbZ_{\tau:\tau+T_c-1}\label{cst:conincideInput}\\
        & \Gamma_{k}^{t} \sqrt{ \Tilde{x}_{t|\tau}^{(l)\intercal}{\Sigma_{k}^{t|\tau}} \Tilde{x}^{(l)}_{t|\tau}} + \Tilde{x}_{t|\tau}^{(l)\intercal}\mu_{k}^{t|\tau} \leq Mz_{k}^{t|\tau (l)}, \notag 
        \\ & \;\;\;\;\;\;\;\;\;\;\;\;\;\;\;\;\;\;\;\;\;\;\;\;\;\;\;\;\;\; \forall t, k \in \cS_{l}, l\in\bbZ_{1:L}, \label{cst:contSafety} \\
        & \eqref{eq:xuconstraints}, \eqref{constraint:RAtj}, \eqref{constraint:RAandBinary}.
    \end{alignat}
\end{subequations}
Note that the constraints \eqref{cst:conincideInput} need to be satisfied only for the first $T_c$ timesteps. In this work, we use $T_c = 1$. 

Solving \eqref{problem:contingencyMPC} generates $L$ trajectories of input $\vu_\tau^{1:L}$ and states $\vx_\tau^{1:L}$. The control input $u_{\tau|\tau}$ and state $x_{\tau+1|\tau}$ that coincides for $T_c = 1$ step. The control inputs for the remaining timesteps $t=\tau+1, \ldots, T_{OH}-1$ may diverge for the $L$ trajectories and will not be used for actuating the EV. The contingency planning scheme plans $L$ trajectories simultaneously, which sacrifices computational efficiency for less conservative decisions.

It is challenging to develop sufficient conditions for recursive feasibility of the contingency planning due to the coinciding constraint \eqref{cst:conincideInput}. Specifically, recursive feasibility requires that a feasible solution at $\tau = 0$ guarantees solutions for all subsequent planning steps $\tau = 1, \ldots, T-1$. However, the $L$ control input trajectories $u_{t|0}^{(l)}$ planned at $\tau = 0$ only coincide for $t=0$ for all $l \in \mathbb{Z}_{1:L}$ and they are permitted to diverge from $t = 1$ onwards. Hence, a feasible coinciding input that satisfies \eqref{cst:conincideInput} does not necessarily exist for the next planning step $\tau=1$. 
However, we can still provide the following safety guarantee if the contingency MPC, i.e., Algorithm \ref{alg:MPCplanning} with $\cP$ set to  \eqref{problem:contingencyMPC}, is recursively feasible and generates a closed-loop trajectory.
\begin{proposition} \textbf{[Contingency planning safety]}
Under Assumption~\ref{assumption:guassianConstraint}, if Algorithm \ref{alg:MPCplanning} with $\cP$ set to \eqref{problem:contingencyMPC} does not return Infeasible and generates a closed-loop state trajectory $\textnormal{\textbf{x}}$, then $\textnormal{\textbf{x}}$ satisfies the original chance constraint \eqref{eq:MPCchance}.
\end{proposition}
\begin{proof}
The proof is a simple consequence that every executed EV state $x_{\tau|\tau-1}$ is a coinciding state for all $L$ trajectories, which satisfies
$$\Gamma_{k}^{\tau} \sqrt{ \Tilde{x}_{\tau|\tau-1}^{\intercal}{\Sigma_{k}^{\tau|\tau-1}} \Tilde{x}_{\tau|\tau-1}} + \Tilde{x}_{\tau|\tau-1}^{\intercal} \cdot \mu_{k}^{\tau|\tau-1} \leq Mz_{k}^{\tau|\tau-1}$$
for all $\tau\in \mathbb{Z}_{1:T}$ and $\forall k \in \bigcup_{l=1}^L \cS_l = \cK$. 
\end{proof}

Each of the above three methods comes with certain advantages and disadvantages. 
 The robust planner guarantees the recursive feasibility under Assumption \ref{assumption:VaRshrink} while it is more conservative than the nominal planner.
 The contingency planner generates the least conservative trajectories at the cost of the longest computational time. To evaluate the performances of these frameworks in practice, we apply these three methods to address autonomous driving scenarios in the next section.

\begin{figure}[!t]
    \centering
    \includegraphics[width=0.5\textwidth]{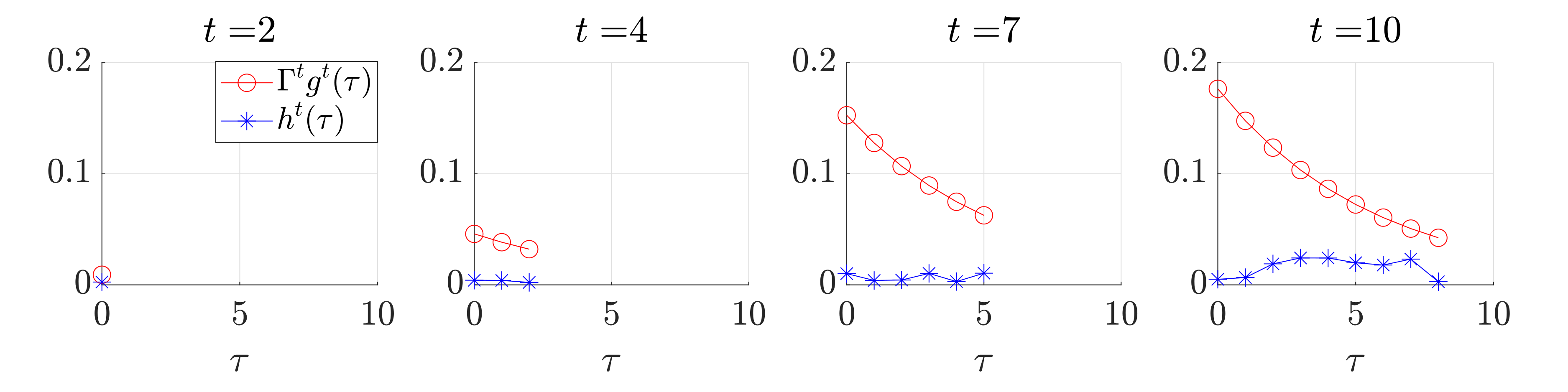}
    \caption{Verification 
    of Assumption \ref{assumption:VaRshrink} for one face of the OV in the lane-changing example. The mean shift $ h^t(\tau)$ in blue is upper bounded by the covariance shrinkage $\Gamma^t g^t(\tau)$ in red. }
    \label{fig:Assump3Curves}
\end{figure}

\begin{figure*}[!ht] 
    \centering
    \subfloat[\revised{Closed-loop trajectories \secrevised{of the EV (blue: the nominal planner, magenta: the robust planner) and the ground-truth poses of the OV (black polytopes).}}]{{\includegraphics[width=\textwidth]{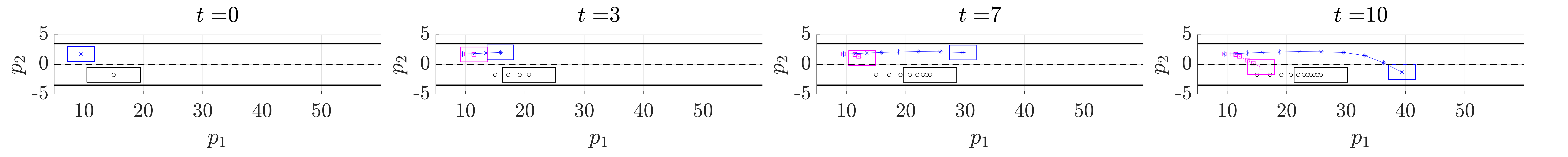} }}\label{fig:yieldclosed}
    \hfill
    \subfloat[\revised{Open-loop trajectories \secrevised{of the EV and the predicted poses of the OV (red polytopes)} generated at planning steps $\tau = 0, 3, 7, 9$.}]{{\includegraphics[width=\textwidth]{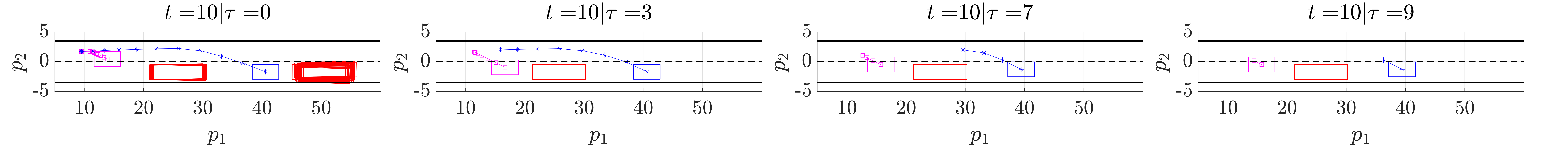} }}\label{fig:yieldopen}
    \caption{OV yield scenario in the lane-changing example. \secrevised{The EV aims at changing to the neighboring lane, while the OV on the neighboring lane slows down and the EV can overtake the OV.
    }}
    \label{fig:yield}
\end{figure*}

\begin{figure*}[!ht] 
    \centering
    \subfloat[\revised{Closed-loop trajectories \secrevised{of the EV and the ground-truth poses of the OV.}}]{{\includegraphics[width=\textwidth]{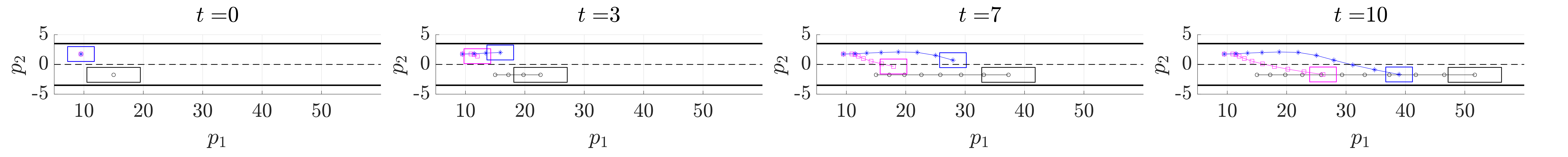} }}\label{fig:notyieldclosed}
    \hfill
    \subfloat[\revised{Open-loop trajectories \secrevised{of the EV and the predicted poses of the OV} generated at planning steps $\tau = 0, 3, 7, 9$.}]{{\includegraphics[width=\textwidth]{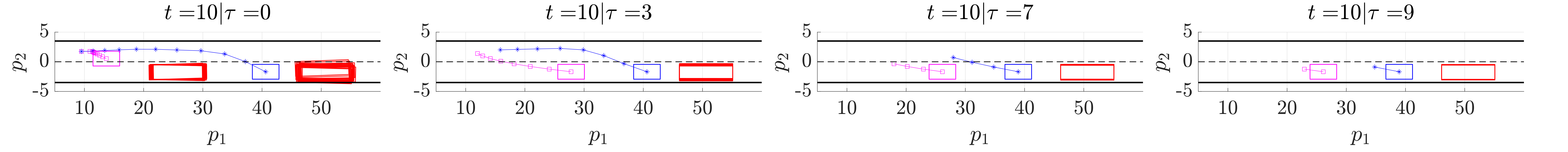} }}\label{fig:notyieldopen}
    \caption{OV accelerate scenario in the lane-changing example.  \secrevised{The EV aims at changing to the neighboring lane, while the OV on the neighboring lane accelerates and the EV cannot overtake the OV. We use the same color coding as in Fig.~\ref{fig:yield}.}}
    \label{fig:notyield}
\end{figure*}

\section{Case study} \label{sec:experiment}
In this section, we demonstrate our theoretical results in two autonomous driving scenarios: lane changing and intersection crossing. The simulations were performed on a personal computer with an AMD Ryzen 5 5600G CPU with six cores at 3.9GHz. As for optimization solvers, we use GUROBI 10.0.1 \cite{gurobi} for lane-changing scenarios in Section~\ref{sec:lane-change} and CPLEX 12.10 \cite{cplex2009v12} for intersection scenarios in Section~\ref{sec:CARLA}.

\subsection{Lane-changing Scenarios in MATLAB}\label{sec:lane-change}

We adopt lane-changing scenarios in \cite{Ahn_2022} where an OV on the targeting lane either slows down (mode 1) to yield or accelerates (mode 2) not to yield. When an EV decides to change a lane, it starts interacting with the OV and begins shrinking-horizon planning as illustrated in Fig.~\ref{fig:CLstatemachine}. 

Let the EV state be ${x} = (p_1, p_2, v_1, v_2)^\intercal$, where $(p_1,p_2)$ is the position of the centroid of the EV and $(v_1, v_2)$ is the velocities along the longitudinal and lateral axis of the lane boundaries, respectively. The EV is modeled as a double integrator:
\begin{equation}\label{eq:doubleInte}
    \dot{{x}} = (\dot{p}_1, \dot{p}_2, \dot{v}_1, \dot{v}_2)^\intercal = (v_1, v_2, u_1, u_2)^ \intercal.
    \end{equation}
Here, the input ${u} = (u_1, u_2)$ is the acceleration along each axis. The acceleration and velocity are bounded as $(u_1, u_2) \in [-10,3] \times [-5,5]$ m/s$^2$ and 
$(v_1, v_2) \in [0, 22.2] \times [-5.56, 5.56]$ m/s. The position $(p_1, p_2)$ is constrained within the boundaries of the road. To describe realistic vehicle behaviors, we impose the following coupling constraints of the acceleration and velocity:
\begin{align} \label{ep:coupled} 
    &\mathcal{U}:=\{u\in\mathbb{R}^2: | u_1 - t_i | + 1.3u_2 \leq 0,\, i=1,2\},
    \\
    &\mathcal{X}:=\{x \in \mathbb{R}^4 |\, | v_1 - c_i | + 2.0v_2 \leq 0,\, i=1,2\},
\end{align}
where $t_1 = -5.56, t_2 = 5.56, c_1 = 0$ and $c_2 = 22.2$. The initial velocity of the EV and the OV is set as 5.56 m/s.
We discretize the model \eqref{eq:doubleInte} using the zero-order hold method with sampling time $0.4$s.

At each planning step $\tau$, we minimize the cost function 
\begin{equation}\label{eq:LaneChangeMPCcost}
    \cJ(\vx_\tau) = \left(x_{T_{OH},2} - x_\text{goal} \right)^2 - 0.1\left(x_{T_{OH},1}\right),
\end{equation}
where $x_\text{goal}$ is the lateral center of the targeting lane. The first term encourages the lane-changing maneuver by minimizing the distance between the EV and the targeting lane. The second term motivates the forward progression of the EV by maximizing the longitudinal traveled distance. 
The overall planning horizon is set to be $T=10$. A shrinking-horizon planning scheme is employed, starting with $T_{OH} = 10$ at the initial planning step $ \tau=0$, followed by $T_{OH} = 9$ at the subsequent step $\tau=1$, and so on. The risk tolerance for the chance constraint is $\epsilon=0.05$.

In this lane-changing test scene, we focus on verifying Proposition \ref{theorem:recursivefeasibility} and evaluating the performance of the nominal and robust planners, that is, Algorithm~\ref{alg:MPCplanning} with $\cP$ set to \eqref{problem:CLCCTPproblem} and \eqref{problem:RobustMPC}, respectively. To validate Proposition \ref{theorem:recursivefeasibility}, we synthesize the prediction distribution such that it satisfies Assumption \ref{assumption:VaRshrink}.
In particular, we let the predicted acceleration of the OV follow a multi-modal Gaussian distribution, which has two modes at the initial planning step $\tau=0$. For the subsequent planning steps $\tau \geq 1$, the number of modes decreases to one. \secrevised{As shown in Fig~\ref{fig:yield}-(b) and \ref{fig:notyield}-(b), the red polytopes represent predictive samples of the OVs. At $\tau=0$, the OV is predicted to exhibit two modes of behavior (i.e., yielding or accelerating). As the planning progresses and the OV's movements are observed, these two modes converge to a single mode.}
\secrevised{To model the decreasing uncertainty in the OV's future states over time, we add Gaussian noise to the initial mean and multiply the initial covariance by a decaying constant of 0.5. This is illustrated by the red polytopes becoming more concentrated as the planning steps progress.}

Fig. \ref{fig:Assump3Curves} shows the mean and covariance shifts over the planning steps $\tau$. The red curves represent the covariance shrinkage $g^t(\tau)$ (and multiplied by $\Gamma^t$) between consecutive planning steps as defined in \eqref{eq:covShrink}, and the blue curves represent the mean shift between consecutive planning steps as defined in \eqref{eq:meanShift}. The fact that the covariance shrinkage multiplied by $\Gamma^t$ is always greater than the mean shift implies that Assumption \ref{assumption:VaRshrink} is satisfied.


The simulation results for OV yielding and accelerating scenarios are shown in Figs. \ref{fig:yield} and \ref{fig:notyield} respectively. We conducted simulations 10 times for each scenario and reported the average performance. The results are summarized in Table \ref{tab:performance-laneChang}. 

\textbf{Recursive feasibility:} In both cases, we observed that the robust planner is recursively feasible as long as it is feasible at $\tau=0$. This confirms Proposition \ref{theorem:recursivefeasibility}. 
We also observe that in the same scenarios, the nominal planner does not return Infeasible and successfully generates a closed-loop trajectory.

\textbf{Optimality:} We evaluate the cost function in \eqref{eq:LaneChangeMPCcost} with the final state of the closed-loop trajectory, which is $\cJ(x_{T})$. The nominal planner generates a smaller cost value than the robust planner. This can also be shown in Figs. \ref{fig:yield} and \ref{fig:notyield}, where the nominal planner (the blue curves) generates less conservative trajectories, i.e., the nominal planner's trajectory tries to cut in between the two predicted OV modes to change the lane and gain more forward progression, whereas the robust planner (the magenta curves) tries to slow down and follows behind all possible OV modes.

\textbf{Computational time:} We report the average worst optimization solver time across each planning step over 10 trials. Note that the two planners take a shorter computation time than the sampling time of $0.4$s, implying the real-time applicability of the proposed planners. 

\begin{table}[b]
    \caption{Performances in the lane-changing scenario} \vspace{-2mm}
    \begin{center} 
    \begin{tabular}{| l | l | r | r |}
        \hline
        \textbf{Method} & \textbf{Scene} & \textbf{Cost} & \textbf{Time (s)}  \\ \hline
        \multirow{2}{*}{\shortstack[l]{Nominal \\
        ($\cP = \eqref{problem:CLCCTPproblem}$)}}  
        & OV yield & -2.56 & 0.15 \\ \cline{2-4}
        & OV accelerate & -2.81  & 0.15   \\ \hline
        \multirow{2}{*}{\shortstack[l]{Robust \\ ($\cP = \eqref{problem:RobustMPC}$)}} 
        & OV yield & 0.46 & 0.05 \\ \cline{2-4}
        & OV accelerate & -2.40 & 0.05\\  \hline
    \end{tabular}
    \end{center}
    \label{tab:performance-laneChang} 
\end{table}

\subsection{Intersection scenarios in CARLA simulator}\label{sec:CARLA}

To test the performance of our methods in real-world autonomous driving scenarios, we utilize the state-of-the-art prediction model \textit{Trajectron++} \cite{Salzmann2020} and run simulations in a realistic autonomous driving CARLA simulator \cite{Dosovitskiy17}. Also, we use ASAM OpenDRIVE that describes the road's features, such as traffic rules in CARLA\footnote{ASAM OpenDRIVE: \url{https://www.asam.net/standards/detail/opendrive}}.

We consider two different scenes:
\begin{itemize}[leftmargin=14pt]
    \item \textbf{T-intersection (T):} The EV interacts with one OV at a T-intersection as shown in Fig. \ref{fig:Tintersection}.
    \item \textbf{Star-intersection (Star):} The EV interacts with two OVs at a star-shaped intersection as shown in Fig. \ref{fig:Starintersection}. 
\end{itemize}
In the simulation, we combine shrinking-horizon planning and receding-horizon planning, as shown in Fig. \ref{fig:CLstatemachine}. Before the EV enters the intersection, we neglect the OV-related chance constraints and employ receding-horizon MPC for lane keeping with an open-loop horizon $T_s=8$ (i.e., 4 seconds) for each MPC step. After the EV enters the intersection, we employ shrinking-horizon MPC with $T=8$, which is required to complete the turn at the intersection. We impose the chance constraints when the EV's planned trajectories are within the intersection. We neglect the chance constraints if the OVs have already exited the intersection.

To consider realistic vehicle motions during a turn, we use a kinematic bicycle model $\dot{x} = f(x, u)$, where
    \begin{equation}\label{eq:bicycle}
    \begin{aligned}
        f(x, u) = \begin{bmatrix}
            \begin{array}{l}
                v \cos(\psi + \gamma) \\
                v \sin(\psi + \gamma) \\
                \frac{v}{L} \cos(\gamma) \tan(u_2) \\
                u_1 \\
            \end{array}
        \end{bmatrix}
    \end{aligned}
\end{equation}

Here, $x = (p_x, p_y, \psi, v)$ and $u = (u_1, u_2)$. For the state, $(p_x, p_y)$ is the position in the global reference frame, $\psi$ is the angle between the vehicle's longitudinal axis and $x$-axis of the global reference frame, and $v \in [0, \revised{10}]$ m/$s$ is the velocity. For the control input, $u_1 \in [\revised{-4}, \revised{7}]$ m/$s^2$ is the acceleration, and $u_2 \in [\revised{-35^\circ},\revised{35^\circ}]$ is the steering angle. Also, $\gamma=\arctan\left( \frac{l_r}{L_{ov}} \tan(u_2) \right)$ is the angle between the vehicle longitudinal axis and velocity vector, $L_{ov}=1$ m is the length of the vehicle, and $l_r=0.5$ m is the length between the back wheel and the vehicle center of gravity. We convert this system to a discrete-time linear time-varying (LTV) system $x_{t + 1}-\bar{x}_{t+1} = A_t (x_t-\bar{x}_t) + B_t (u_t-\bar{u}_t)$ by linearizing \eqref{eq:bicycle} around a nominal trajectory $(\bar{\vx}, \bar{\vu})$ \cite{Falcone2007}. 

\begin{figure}[!t] 
    \centering
    \includegraphics[width=0.4\textwidth]{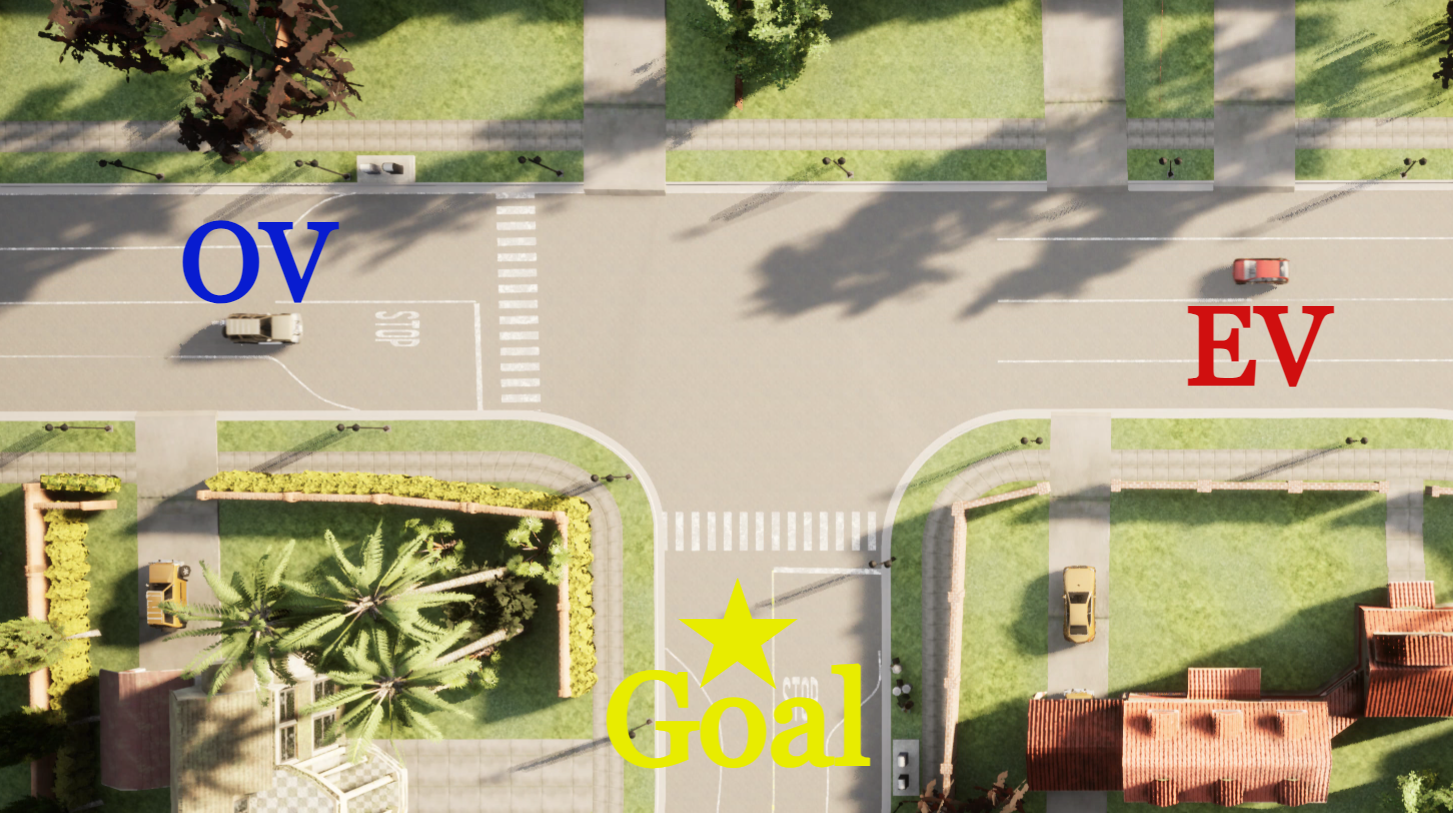}
    \caption{T-intersection scene in the CARLA simulator. The EV (red car) intends to make a left turn and reach the yellow star goal in the presence of one OV (white car).}
    \label{fig:Tintersection}
\end{figure}
\begin{figure}[!t] 
    \centering
    \includegraphics[width=0.4\textwidth]{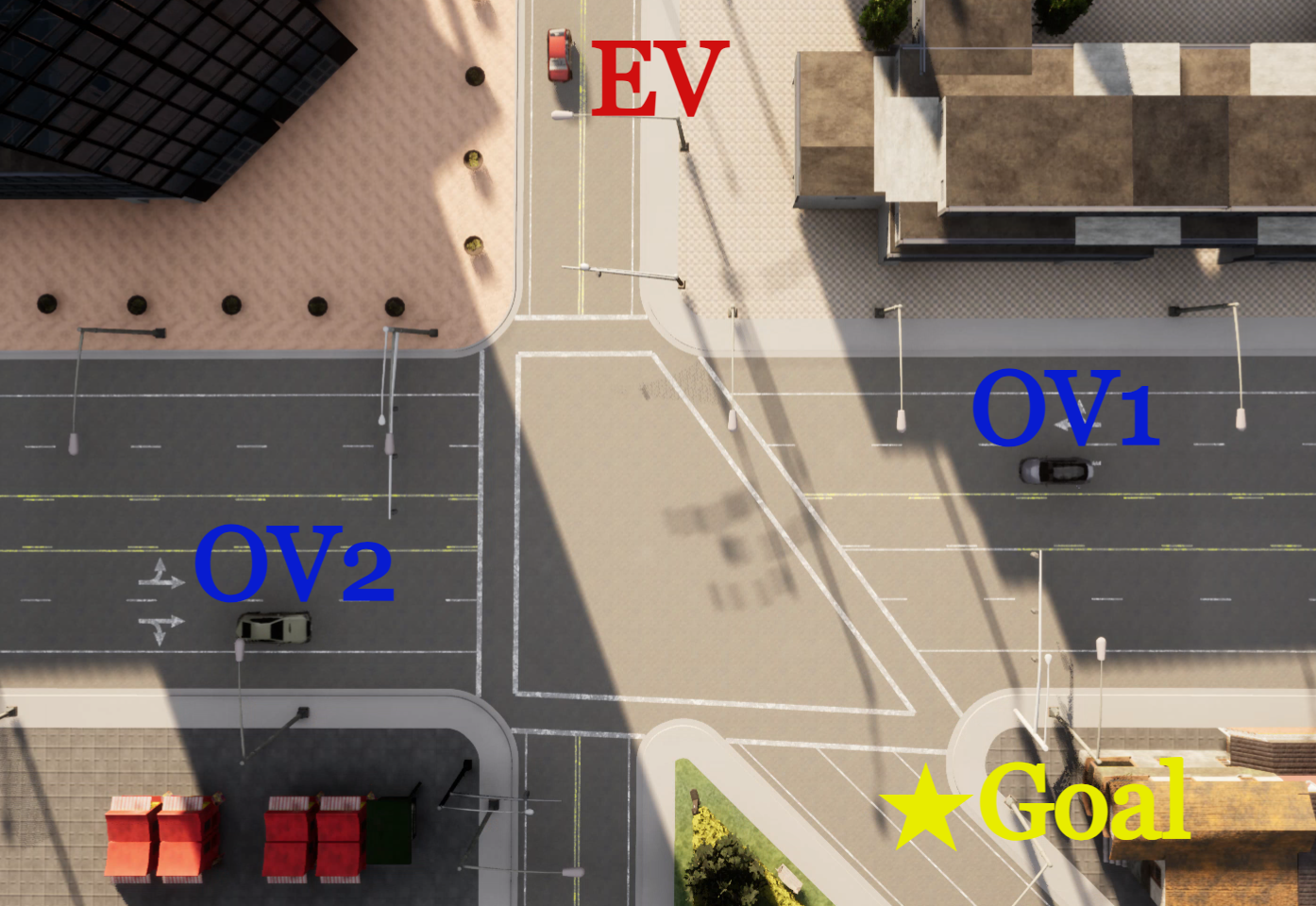}
    \caption{Star-intersection scene in the CARLA simulator. The EV (red car) intends to reach the yellow star goal. OV$_1$ (gray car on the right) turns to the left, and OV$_2$ (gray car on the left) turns to the right at the intersection.}
    \label{fig:Starintersection}
\end{figure}

\begin{figure*}[!ht] 
    \centering
    \captionsetup[subfloat]{labelformat=empty}
    \subfloat[t = 2]{{\includegraphics[width=0.185\textwidth]{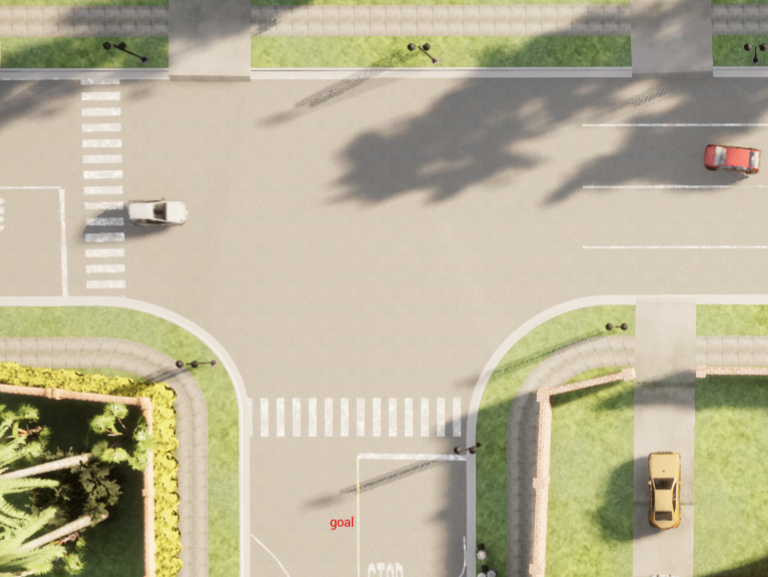} }}\label{fig:T_0} 
    \subfloat[t = 4]{{\includegraphics[width=0.185\textwidth]{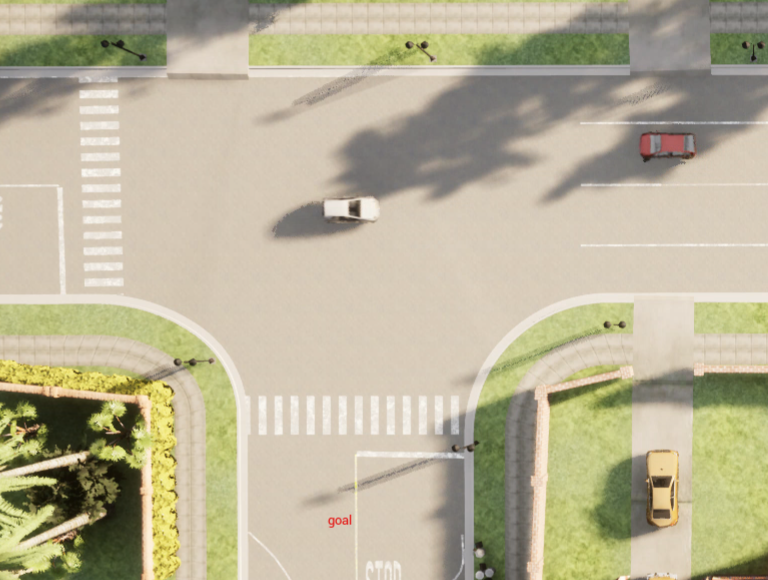} }}\label{fig:T_2}
    \subfloat[t = 6]{{\includegraphics[width=0.185\textwidth]{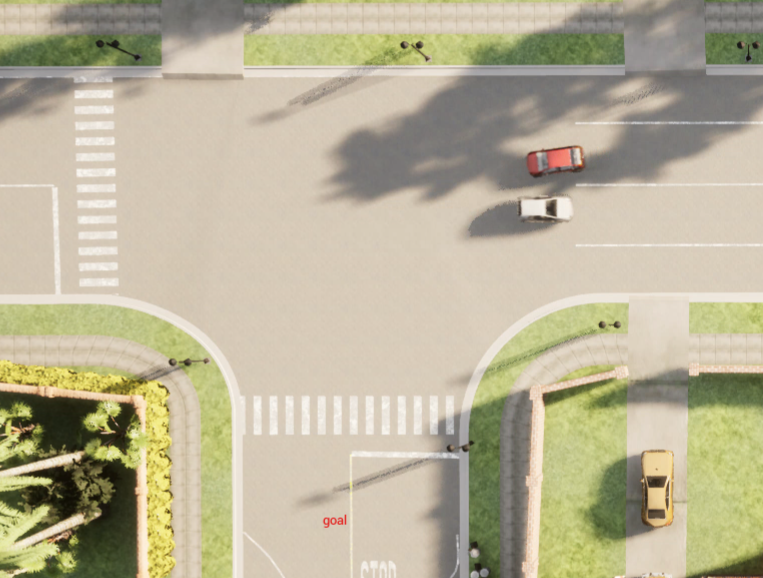} }}\label{fig:T_4}
    \subfloat[t = 8]{{\includegraphics[width=0.185\textwidth]{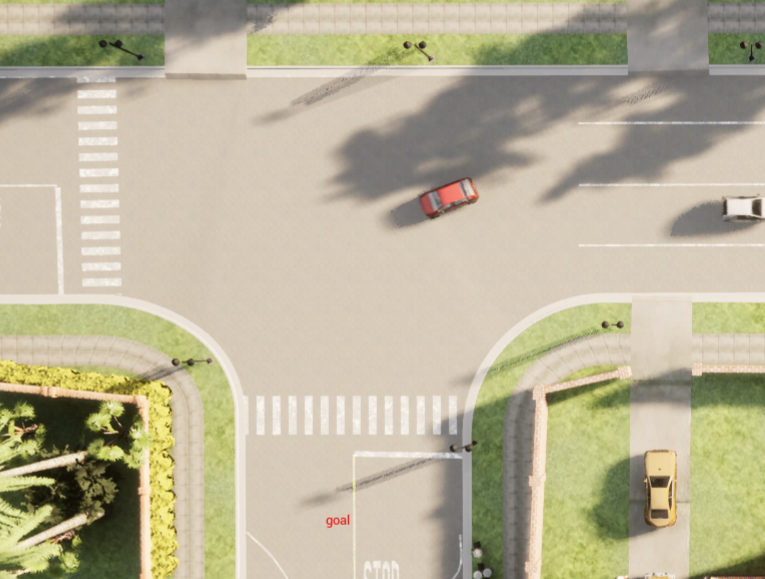} }}\label{fig:T_6}
    \subfloat[t = 10]{{\includegraphics[width=0.185\textwidth]{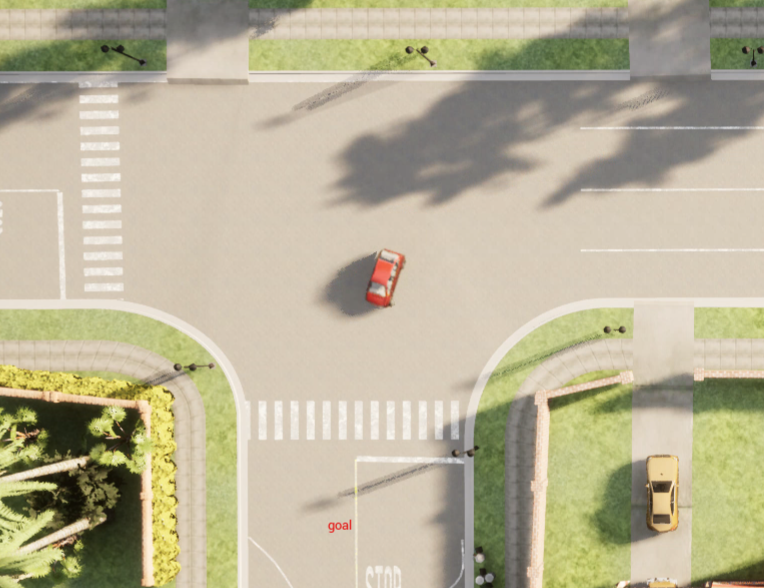} }}\label{fig:T_8}
    \caption{Simulation of the nominal planner in the T-intersection scene. The EV's trajectory (red car) is generated by the receding-horizon and shrinking-horizon implementation, and shrinking-horizon planning starts at $t=2$.
    } 
    \label{fig:nominal_T_snapshot}
\end{figure*}

\begin{figure*}[!ht] 
    \centering
    \captionsetup[subfloat]{labelformat=empty}
    \subfloat[t = 0]{{\includegraphics[width=0.185\textwidth]{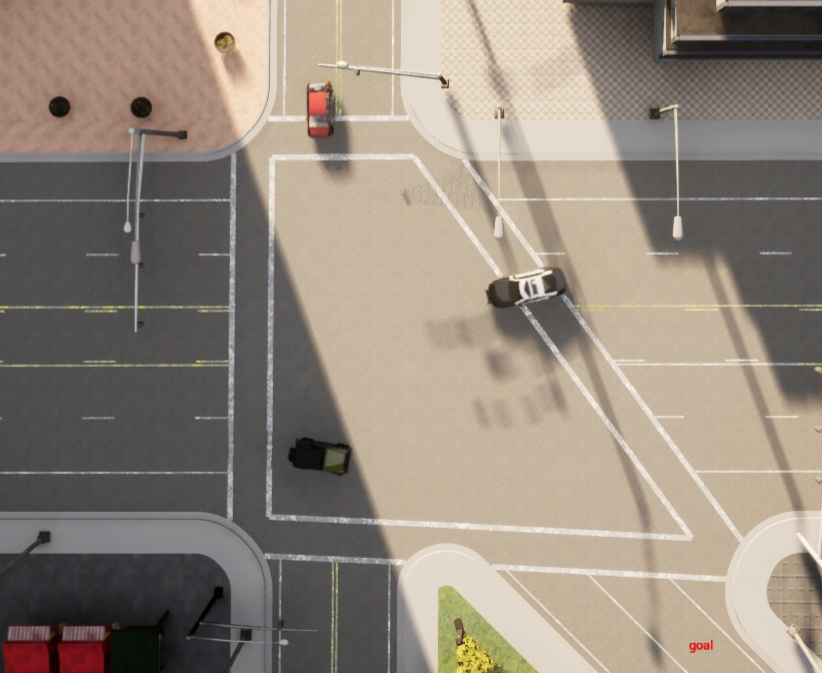} }}\label{fig:Star_0}
    \subfloat[t = 2]{{\includegraphics[width=0.185\textwidth]{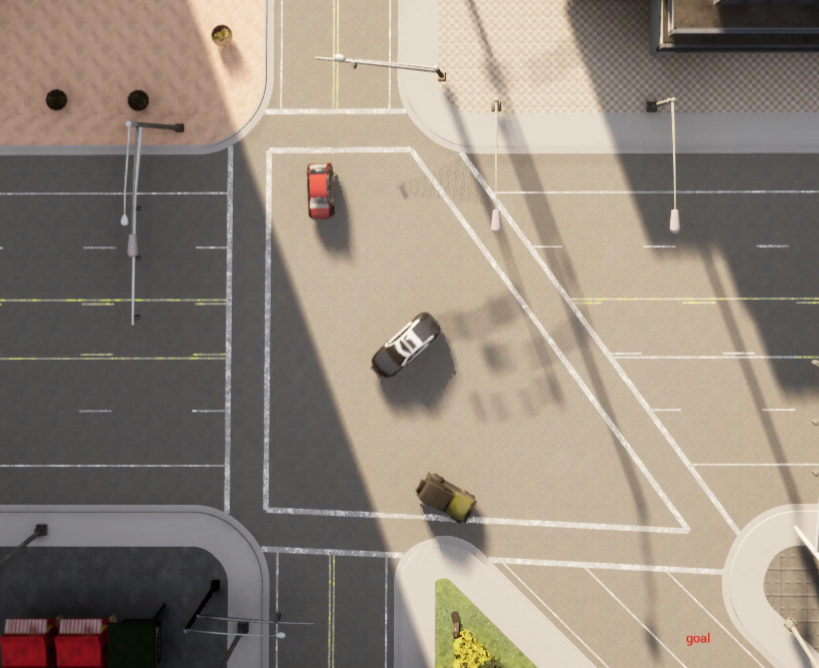} }}\label{fig:Star_2}
    \subfloat[t = 4]{{\includegraphics[width=0.185\textwidth]{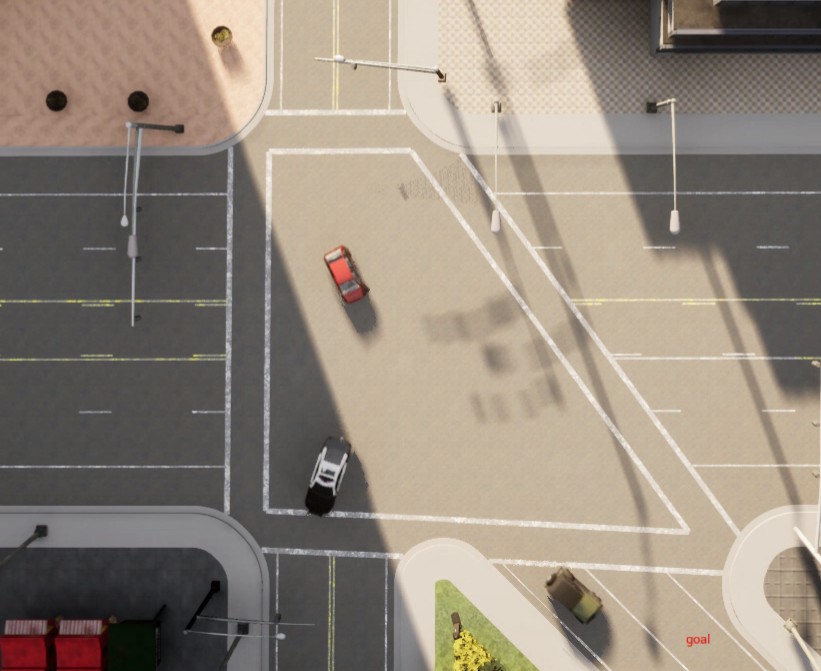} }}\label{fig:Star_4}
    \subfloat[t = 6]{{\includegraphics[width=0.185\textwidth]{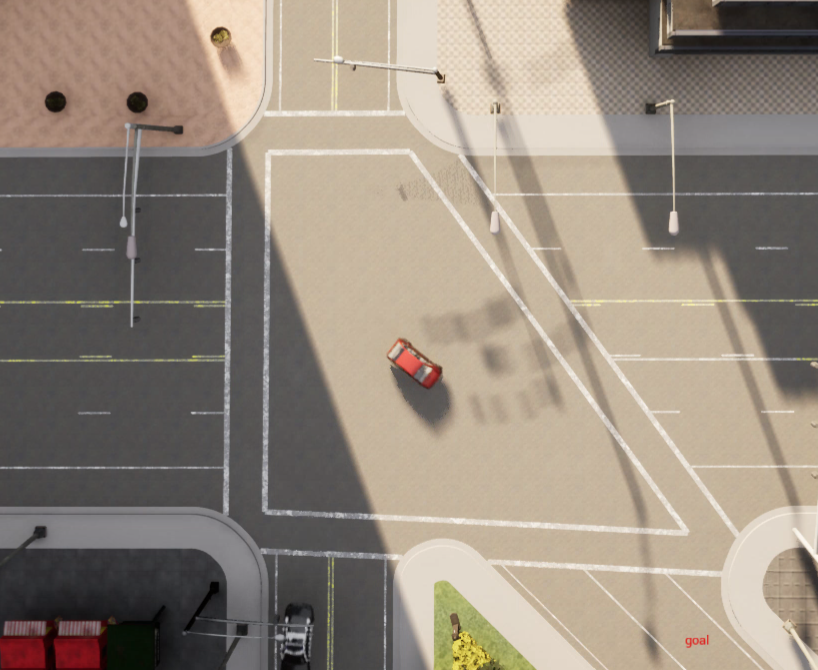} }}\label{fig:Star_6}
    \subfloat[t = 8]{{\includegraphics[width=0.185\textwidth]{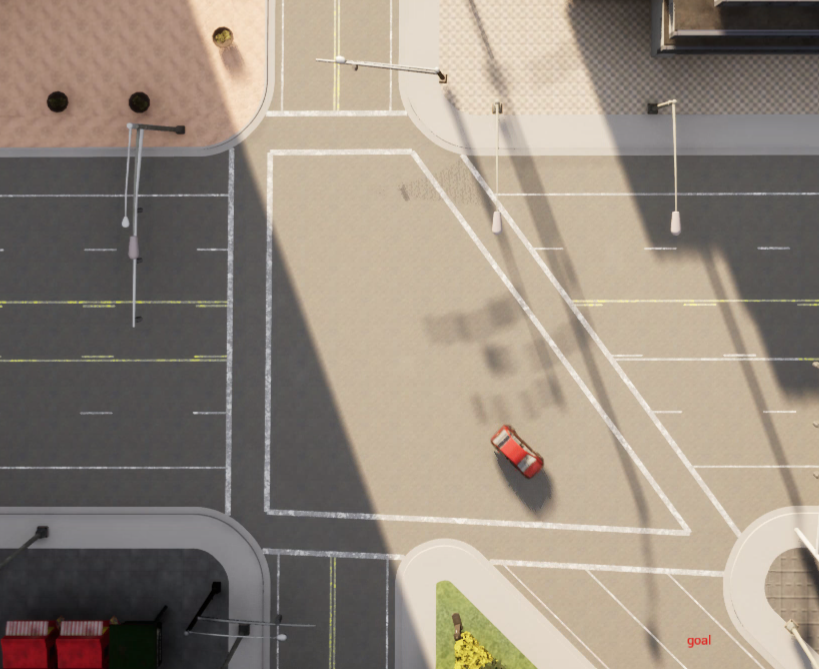} }}\label{fig:Star_8}
    \caption{Simulation of the nominal planner in the star-intersection scene. Shrinking-horizon planning starts at t=0. 
    }
    \label{fig:nominal_Star_snapshot}
\end{figure*}


At each planning step $\tau$, we minimize the following quadratic cost function
\begin{equation}\label{eq:MPCcost}
    \cJ(\vu_\tau, \vx_\tau) = \norm{ {x}_{T_{OH}} - {x}_\text{goal} }_P^2
    + \sum^{T_{OH} -1}_{t=0} \left(\norm{ u_t }_{R_1}^2 + \norm{\Delta u_t}_{R_2}^2 \right),
\end{equation}
where $\Delta u_t = u_{t+1} - u_t$ and 
$P = \begin{bmatrix}
300.0 & 0\\
0 & 300.0
\end{bmatrix}, \;\; R_1 = \begin{bmatrix}
0.05 & 0.02\\
0.02 & 0.10
\end{bmatrix}, \;\; R_2 = \begin{bmatrix}
0.05 & 0.01\\
0.01 & 0.20
\end{bmatrix}$.
The first term of the objective function minimizes the terminal state distance to the goal state. To promote smooth and stable driving behavior, the second term minimizes both the magnitude of acceleration and steering interventions and the differences in control inputs across consecutive time steps.

 
The prediction model \textit{Trajectron++} \cite{Salzmann2020} generates \(N_s\) samples of OV trajectory predictions. Each sample has an associated discrete latent variable, which encodes an associated GMM mode. We determine the number of modes based on the discrete latent variables' values that exceed a certain threshold; for example, $K=2$ at the beginning of shrinking-horizon planning in the T-intersection scene, and $K_1=3, K_2=3$ at the beginning of shrinking-horizon planning in the star-intersection scene. Then, we cluster the samples according to their latent variable's value, obtaining $N_k$ samples for each mode $k$ and estimating the GMM moments from the $N_k$ samples. The risk tolerance for the chance constraint \eqref{eq:MPCchance} is $\epsilon=0.05$.

 Under the predictions from \textit{Trajectron++}, Assumption \ref{assumption:VaRshrink} does not always hold. 
 The predictions corresponding to OV-turning exhibit substantial mean shifts for consecutive time steps, making the fulfillment of Assumption~\ref{assumption:VaRshrink} impractical. Also, the covariance of predictions occasionally exhibits irregular increases or decreases in magnitude. These violations make it difficult to validate Proposition \ref{alg:MPCplanning} with the prediction model. Improving the predictions' quality can help alleviate these violations, which could be considered as future work.

 We implement the nominal, robust, and contingency MPC planners, that is, Algorithm~\ref{alg:MPCplanning} with $\cP$ set to \eqref{problem:CLCCTPproblem}, \eqref{problem:RobustMPC}, and \eqref{problem:contingencyMPC}, respectively. For contingency planning, we use a coinciding horizon $T_c = 1$ and plan $L$ trajectories. Each plan accounts for different combinations of chance constraints. For instance, when OV$_1$ has three modes and OV$_2$ has two modes, we plan $L=3$ trajectories where $\vu^{(1)}$ handles the chance constraints of the first modes of OV$_1$ and OV$_2$, $\vu^{(2)}$ the second modes of OV$_1$ and OV$_2$, and $\vu^{(3)}$ the third mode of OV$_1$.  
 
 We conducted 100 trials of simulations for each scenario; in each trial, shrinking-horizon planning is initially feasible. In each trial, the OVs start from the same initial states, yet their accelerations, steering inputs, and intentions (e.g., going straight or turning right at the T-intersection) vary in different trials. One exemplary trial of the nominal planner for the T-intersection scene is shown in Fig.~\ref{fig:nominal_T_snapshot} and for the star-intersection scene in Fig.~\ref{fig:nominal_Star_snapshot}. We compare the performance of the three planners introduced in Section~\ref{section:unifiedPlanning} with several criteria, and the results are summarized in Table \ref{tab:performance}.


\textbf{Open-loop \& closed-loop planning:} We first evaluate the improvement brought by employing our MPC framework in comparison to open-loop planning. In particular, the open-loop planner tries to generate a trajectory by solving the trajectory planning problem just once at the very initial planning step, i.e., solving either nominal planning \eqref{problem:CLCCTPproblem} or contingency planning \eqref{problem:contingencyMPC} with $\tau = 0$ and $T_{OH} = T$ that is the required planning horizon length to reach the goal. The open-loop planning scheme cannot find a feasible solution in all scenes due to the highly uncertain OV prediction over a long planning horizon (as shown in the last row of Table~\ref{tab:performance}). On the other hand, the MPC framework overcomes this infeasibility by gaining new observations of the OVs' motion and correspondingly updating the OV predictions and EV trajectories at each planning step.

\textbf{Feasibility:} As for the feasibility criteria, we count the number of trials where each planner returns Infeasible before the EV reaches the goal. For example, if the feasibility is 93\% (nominal planner in the T-intersection scene), it indicates that the nominal planner generates a closed-loop trajectory in 93 trials out of the total 100 trials. 

We observe that both the nominal and contingency planners have a high feasibility rate in the two scenes, while the robust planner has a lower feasibility rate. As we noted earlier in this section, the robust planner does not have a recursive feasibility guarantee because the predictions from \textit{Trajectron++} do not satisfy Assumption~\ref{assumption:VaRshrink}.
\revised{The contingency planner has a higher feasibility rate than the nominal planner in both scenes. 
In cases where the contingency planner returns infeasible, we observe that the EV tends to make aggressive maneuvers in the early steps of MPC and fails to find a feasible trajectory in the subsequent steps.}
The robust planner cannot find a feasible closed-loop solution in the star-intersection scene due to the high interaction and highly uncertain predictions of the OVs. 



\textbf{Optimality:} As for the optimality criteria, we evaluate the average cost and the traveled time. The cost value of each feasible trial is calculated based on \eqref{eq:MPCcost} over the planning horizon until the EV reaches the goal. In the column of ``\textbf{Cost}" of Table~\ref{tab:performance}, we report the average cost values over feasible trials. The travel time refers to the average time taken for the EV to reach the goal from its initial state. We summarize the average travel time among the feasible trials in the column of ``\textbf{Travel time}" of Table~\ref{tab:performance}. Our objective is to minimize the cost, and typically, a lower cost results in fewer planning steps needed to reach the goal. 

In both scenes, the contingency planner has the best performance in terms of minimizing the cost, followed by the nominal planner. The robust planner generates trajectories with much higher cost and travel time. For example, at the T-intersection, the EV driven by the robust planner fully stops to wait for the OVs to exit the intersection, thereby taking a long time to reach the goal, whereas the EV only slows down with the nominal and contingency planners.

\revised{\textbf{Empirical collision rate:} For the safety criteria, we evaluate a collision rate based on $10^4$ new OV parameters sampled from the GMM distribution. The data is sampled from the time step corresponding to the EV trajectory planning. For instance, 
to verify the safety of $x_{1|0}$, we generate samples from \((\mu_{ijk}^{1|0}, \Sigma_{ijk}^{1|0}), \; \forall i, j, k\) to empirically measure the probability of collisions. We report the average collision rate over the feasible trials in the column of \textbf{Coll. Rate} of Table~\ref{tab:performance}. We observe that whenever a feasible closed-loop EV trajectory can be found, the probability of constraint violation is 0\%, despite setting a constraint violation threshold of $\epsilon=5\%$. This outcome aligns with our expectations, as the deterministic reformulations \eqref{eq:chanceGMM}, \eqref{eq:worstVaRconstraint}, and \eqref{cst:contSafety} are conservative approximations of the original chance constraints \eqref{eq:MPCchance}.
}

\begin{table}[b]
    \caption{Performances in the intersection scenarios} \vspace{-2mm}
    \begin{center} 
    \begin{tabular}{| l | l | r | r | r | r | r |}
        \hline
        \textbf{\shortstack[c]{\vspace{1.5mm} \; Method}} & \textbf{\shortstack[c]{\vspace{1.6mm}Scene}} & \textbf{\shortstack[c]{\vspace{1mm}Feasibility}} & \textbf{\shortstack[c]{Travel \\
        time}} & \textbf{\shortstack[c]{\vspace{1.5mm}Cost}} & \revised{\textbf{\shortstack[c]{Coll. \\
        Rate}}} & \textbf{\shortstack[c]{Comp. \\
        time}}  \\ 
        \hline
        \multirow{2}{*}{\shortstack[l]{Nominal \\
        ($\cP = \eqref{problem:CLCCTPproblem}$)}}  
        & T    & \revised{93\%} & \revised{6.88s} & \revised{4587} & \revised{0 \%} & \revised{3.05s} \\ \cline{2-7}
        & Star & \revised{92\%} & \revised{5.20s}  & \revised{5588} & \revised{0 \%} & \revised{4.65s} \\ \hline
        \multirow{2}{*}{\shortstack[l]{Robust \\
        ($\cP = \eqref{problem:RobustMPC}$)}}  
         & T  & \revised{85\%} & \revised{8.32s} & \revised{9394} & \revised{0 \%} & \revised{2.89s} \\ \cline{2-7}
        & Star& 0\% & - & -  & - & - \\ \hline
        \multirow{2}{*}{\shortstack[l]{Contingency \\
        ($\cP = \eqref{problem:contingencyMPC}$)}}  
         & T  & \revised{97\%} & \revised{6.62s} & \revised{4531} & \revised{0 \%} & \revised{14.1s} \\ \cline{2-7}
        & Star& \revised{95\%} & \revised{5.10s} & \revised{4839} & \revised{0 \%} & \revised{6.84s} \\ \hline \hline
        \rule{0pt}{10pt} \shortstack[l]{\hspace{-1mm}Open-loop \vspace{0.5mm}} & \shortstack{\vspace{1mm}All} & \shortstack{\vspace{1mm} 0\% } & \shortstack{\vspace{1mm} -} & \shortstack{\vspace{1mm} -} & \shortstack{\vspace{1mm} -} & \shortstack{\vspace{1mm} -}  \\ \hline
    \end{tabular}
    \end{center}
    \label{tab:performance} 
\end{table}



\textbf{Computational time:} To evaluate how close our methods are to real-time applicability, we measure the worst optimization solver time taken to solve $\cP$ of Algorithm \ref{alg:MPCplanning} in each trial. We report the average worst time over 100 trials in the column of ``\textbf{Comp. time}" in Table~\ref{tab:performance}. The contingency planner has the longest computational time as it plans $L$ trajectories simultaneously.

Given that we solve the trajectory planning problem at a frequency of every 0.5\,s, the computational time of our planners cannot satisfy the real-time applicability. In comparison with the previous chance-constrained MPC framework under GMM uncertainty in \cite{Nair2022}, which offers real-time applicability, a key difference is that in \cite{Nair2022}, the lane-keeping behavior is tackled by planning a trajectory that closely aligns with a reference trajectory that is computed offline. In contrast, our approaches incorporate lane-keeping as a constraint, where we force the state to stay within (curved) lane boundaries described by the union of convex sets. \revised{Moreover, we considered polyhedral obstacles for collision avoidance, which necessitated additional integer variables to reformulate the disjunctive constraints. An optimization-based method \cite{Zhang2021} or locally linear approximations of the ellipsoidal obstacle constraints \cite{Nair2022} can reduce the integer variables required for collision avoidance and further decrease the computational cost.}
We believe that the computation time can be further reduced by reformulating the lane-keeping \revised{and collision avoidance} constraints and optimizing the codes, which remains as a future work.

In summary, the realistic simulations indicate that our proposed methods can ensure safety while conservativeness and/or computation time can be a tradeoff. In particular, the robust planner offers a provable recursive feasibility guarantee but tends to generate conservative trajectories. The less conservative planning of the contingency planner comes at the price of increased computation time. The nominal planner appears to offer both recursive feasibility and safety throughout our simulation trials.

\section{Conclusion}
We presented a chance-constrained MPC framework that considers GMM uncertainty. With assumptions on the propagation of the uncertainty's GMM prediction, we ensured that a robust MPC planner is recursively feasible and safe with respect to the chance constraint throughout the entire planning horizon. A less conservative chance-constrained MPC planner was designed based on a contingency planning method, which reduces the conservativeness under GMM uncertainty by planning multiple trajectories simultaneously. Through a case study using a cutting-edge autonomous driving simulator, we observed that the MPC framework removes the conservatism of open-loop controllers in finite-horizon trajectory planning problems. While the contingency planning approach enhanced the trajectory's optimality, it came at the expense of extended computational time. Our planning framework requires further improvement in terms of computation time to support real-time applications. Moreover, enhancing the trajectory prediction quality of other vehicles can reduce the conservativeness of risk-constrained problems.

\bibliographystyle{IEEEtran}
\bibliography{main}

\appendices

\revised{\section{Selection of Risk-constrained Formulations \secrevised{for a Simple Example}} \label{Appendix:risk-constrained}}
\revised{
 We investigate six different risk-constrained methods assuming the uncertainty's distribution is fully known, partially known, or completely unknown. We evaluate the performances of the risk-constrained approaches based on the following numerical risk-constrained optimization problem. }
\begin{subequations} \label{simulation:numerical}
\begin{alignat}{2}
    & \underset{x \in \mathbb{R}}{\text{min}} && \;\;\; x, \\ 
    & \text{subject to} && \;\;\; \mathbb{P}(\delta \leq x) \geq 1-\epsilon, \label{eq:single-risk}
\end{alignat}
\end{subequations}
where $x\in\mathbb{R}$ is a decision variable and $\delta\in \mathbb{R}$ is a random variable with a GMM distribution $\delta \sim \sum_{k=1}^{K} \pi_k \cdot \cN (\mu_{k}, \Sigma_{k})$. For the risk metric, we consider either the original chance constraint \eqref{eq:single-risk} (abbreviated as CC) or a CVaR constraint, which is defined as
        \begin{equation}\label{simulation:cvarnumerical}
        {\text{CVaR}_{\epsilon}}\left(\delta - x\right) := 
        \inf_{\alpha>0} \left\{ \frac{1}{\epsilon} \mathbb{E}[(\delta-x+ \alpha)_+] - \alpha \right\} \leq 0.
        \end{equation}

 Note that the CVaR constraint \eqref{simulation:cvarnumerical} is a conservative approximation of the chance constraint \eqref{eq:single-risk} \cite{Rockafellar2001}.

First, we consider the case when the exact GMM moments are known. The following risk-constrained methods can be employed to address uncertainty.
\begin{enumerate}[leftmargin=26pt, label=(M\arabic*)]
    \item \label{num:MTA} CC-Moment trust approach (MTA) \cite{Hu2022}: When the exact GMM moments are known, the chance constraint \eqref{eq:single-risk} can be equivalently formulated as the following constraints.
    \begin{equation*} \label{cha1:eq:numericalBi}
        \begin{cases}
        \Psi^{-1}(1-\epsilon_k) {\sigma}_k + {\mu}_k \leq x, \;\; \forall k\\
        \sum_k \pi_k \cdot \epsilon_k= \epsilon.
        \end{cases}
    \end{equation*}

    \item \label{num:ICP} CVaR-MTA \cite{You2021GMMCVaR}: 
        When the exact GMM moments are known, an iterative cutting plane algorithm \cite{You2021GMMCVaR} is designed to converge to a solution $x^*$ that satisfies \eqref{simulation:cvarnumerical}.   
\end{enumerate}

 Second, we consider the situations where $\delta$ is known to follow a GMM with $K$ modes, but the precise GMM moments are unknown, and only $N_s$ independent and identically distributed (i.i.d.) samples of $\delta$ are accessible from the true GMM distribution, denoted as $\delta^{(1)}, \ldots, \delta^{(N_s)}$. Given the number of modes, we can also determine mode $k$ where each sample belongs based on clustering methods \cite{bishop2006}. We use $\delta_k^{(n)}$ to denote a sample number $n$ that belongs to mode $k$. In this case, the following methods can be applied to manage risk under uncertainty.

  \begin{enumerate}[leftmargin=26pt, label=(M\arabic*)]
    \setcounter{enumi}{2}
    \item \label{num:MRA} CC-Moment robust approach (MRA) \cite{Ren2023}: The moments of the $k^{th}$ GMM mode can be estimated from the samples. MRA robustifies the estimated moments $(\hat{\mu}_k,\hat{\sigma}_k)$ and formulates
    \begin{equation*}
        \begin{cases}
        \Psi^{-1}(1-\epsilon_k) \hat{\sigma}_k \sqrt{1+r_{2,k}} + \hat{\mu}_k + r_{1,k} \leq x, \;\; \forall k \\
        \sum_k \pi_k\cdot \epsilon_k = \epsilon,
        \end{cases}
    \end{equation*}
    where $r_{1,k}$ and $r_{2,k}$ quantifies the error bounds between the estimate and actual moments (see \cite{Ren2023}). This approach provides a guarantee of the satisfaction of the chance constraint \eqref{eq:single-risk} with $1-2\beta$ probability, where $\beta$ is a prescribed tolerance.

    \item \label{num:scenario} CC-scenario approach \cite{Ahn_2022}: With the prior information on the number of modes and the mode $k$ where each sample belongs, the scenario approach formulates
    \[
    \begin{aligned}
        x \geq \max_{i} \left(\delta_k^{(i)} \right), \;\;\; \forall k. \\
    \end{aligned}
    \]
    The solution of this constraint is guaranteed to satisfy the chance constraint \eqref{eq:single-risk} with at least $1-\beta$ probability, where $\beta$ is determined by the prediction sample size $N_s$ and chance constraint risk bound $\epsilon$.
\end{enumerate}

 Third, we consider the case when there is no prior information on the distribution of $\delta$ and only $N_s$ i.i.d. samples of $\delta$ drawn from the true distribution are available. In this case, the following CVaR-constrained approaches can be employed. 

\begin{table*}[!ht]
\centering
\begin{minipage}{\textwidth}
    \caption{Summary of risk-constrained approaches under GMM uncertainty}
    \begin{center} \vspace{-2mm}
    \begin{tabular}{| l | l | l | r | l | c | c |} 
        \hline
       \rule{0pt}{18pt}  & \thead{\shortstack[l]{\vspace{1.5mm}\textbf{Risk measure}}} & \thead{\shortstack[l]{\vspace{1.5mm}\textbf{Method}}} & \thead{\shortstack[c]{\vspace{1mm}\textbf{Comp. time (s)}}} & \thead{\textbf{\shortstack[l]{\vspace{1mm}Risk constraint guarantee}}} \\ 
        \hline 
        \multirow{2}{*}{\shortstack[l]{Known \\ GMM}} 
        & \multirow{1}{*}{\shortstack[c]{Chance}} 
        & 
        \ref{num:MTA}
        Moment Trust Approach \cite{Hu2022} & $2.60 \times 10^{-1}$ 
        & Deterministic  \\
        \cline{2-5}
        & \multirow{1}{*}{\shortstack[c]{CVaR}}  
        & 
        \ref{num:ICP}
        Iterative Cutting-plane \cite{You2021GMMCVaR} & $6.39 \times 10^{0} $ 
        & None \\
        \cline{2-5}
        \hline\hline
        \multirow{5}{*}{\shortstack[l]{Data-driven \\ Case}}  
        & \multirow{2}{*}{\shortstack[c]{Chance}} 
        & 
        \ref{num:MRA} 
        Moment Robust Approach \cite{Ren2023} & $2.47 \times 10^{-1}$ 
        & Probabilistic \\ 
        \cline{3-5}
        & & 
        \ref{num:scenario} 
        Scenario Approach\cite{Ahn_2022} & $4.28 \times 10^{-5}$ 
        & Probabilistic \\ 
        \cline{2-5}
        & \multirow{2}{*}{\shortstack[c]{CVaR}} 
        & 
        \ref{num:SAA} 
        Sample Average Approximation \cite{Hakobyan2019} & $3.81 \times 10^{0} $  
        & None \\
        \cline{3-5}
        & & \ref{num:WassDR}
        Distributionally Robust \cite{Esfahani2015DRO} & $3.00 \times 10^{0} $ 
        & Probabilistic \\
        \hline
    \end{tabular}
    \end{center}
    \label{tab:OLTPmethodsSummary}
\end{minipage}
\end{table*} 

\begin{enumerate}[leftmargin=26pt, label=(M\arabic*)]
\setcounter{enumi}{4}

    \item \label{num:SAA} CVaR-sample average approximation (SAA) method \cite{Hakobyan2019}: With $N_s$ samples of $\delta$, the SAA method estimates the CVaR constraint as
    \begin{subequations} \label{problem:CVaRCPSAA}
    \begin{alignat*} {2}\frac{1}{\epsilon N_s} \sum_{i=1}^{N_s} [(\delta^{(i)} - x +\alpha)_+] - \alpha \leq 0.  \label{constraint:CVaRSAA}
    \end{alignat*}
\end{subequations}

    \item \label{num:WassDR} CVaR-distributionally robust (DR) approach \cite{Esfahani2015DRO} with Wasserstein-distance ambiguity set: The Wasserstein-distance DR CVaR constraint is formulated as
    \[
    \begin{aligned}
        \underset{Q \in \mathcal{Q}_\gamma(\hat{P}_{N_s})}{\text{sup}} \text{CVaR}_{\epsilon}^Q(\delta - x) \leq 0,
    \end{aligned}
    \]
    where $\hat{P}_{N_s}$ is an empirical distribution of $\delta$ estimated with $N_s$ samples of $\delta$. Also, $\mathcal{Q}_\gamma(\hat{P}_{N_s})$  is the set of distributions whose Wasserstein distance from $\hat{P}_{N_s}$ is less than or equal to $\gamma$. This approach provides a satisfaction guarantee of the chance constraint \eqref{eq:single-risk} with at least $1-\beta$ probability, where $\beta$ depends on the Wasserstein radius $\gamma$, sample size $N_s$ and the underlying distribution of $\delta$.

\end{enumerate}


To compare the optimality, safety and computation time, we evaluate the above-mentioned risk-constrained methods on problem \eqref{simulation:numerical}, subject to uncertainty with a bi-modal Gaussian distribution $\delta \sim 0.5\mathcal{N}(1, 1)+0.5\mathcal{N}(10, 1)$. 

We use the approaches \ref{num:MTA}-\ref{num:WassDR} to solve the above risk-constrained problem for 100 times. At each time, we draw $N_s=2000$ i.i.d. samples of $\delta$ from the GMM, and we cluster the samples into two modes and estimate the GMM moments $\{(\hat{\mu}_k, \hat{\sigma}_k)\}_{k=1}^2$. When the GMM moments are assumed to be known, particularly for methods \ref{num:MTA} and \ref{num:ICP}, we trust the sample-estimated GMM moments. For \ref{num:MRA}, we robustify against the moment estimation error. For methods \ref{num:scenario}, \ref{num:SAA} and \ref{num:WassDR}, we directly use the samples to solve the risk-constrained problem. In each simulation, we obtain the optimal solutions $x^*$ based on the above-mentioned methods. After obtaining the optimal solution $x^*$, we evaluate the constraint violation probability based on the $10^4$ new test samples of $\delta$ drawn from the true distribution, i.e., we count the rate of $\delta^n_{\text{test}} \geq x^*$ among the new test samples $n = 1, \ldots, 10^4$. 

We set the risk tolerance to be $\epsilon = 0.05$. In methods \ref{num:MTA} and \ref{num:MRA}, where we need to assign the risk to different GMM modes, we uniformly assign $\epsilon_k = \epsilon$ for all $k\in \bbZ_{1:K}$. For the CVaR-DR method \ref{num:WassDR}, we choose a Wasserstein radius of the ambiguity set to be $\gamma = 0.04$. 

Note that problem \eqref{simulation:numerical} aims at minimizing $x$, and a larger solution is safer to the risk constraint, as $\mathbb{P} (\delta \leq x)$ increases with the value of $x$. Hence, ``less conservative" and ``lower cost" will interchangeably describe a smaller $x$.

\begin{figure}[!t] \vspace{-3mm}
    \centering
    \includegraphics[width=0.45\textwidth]{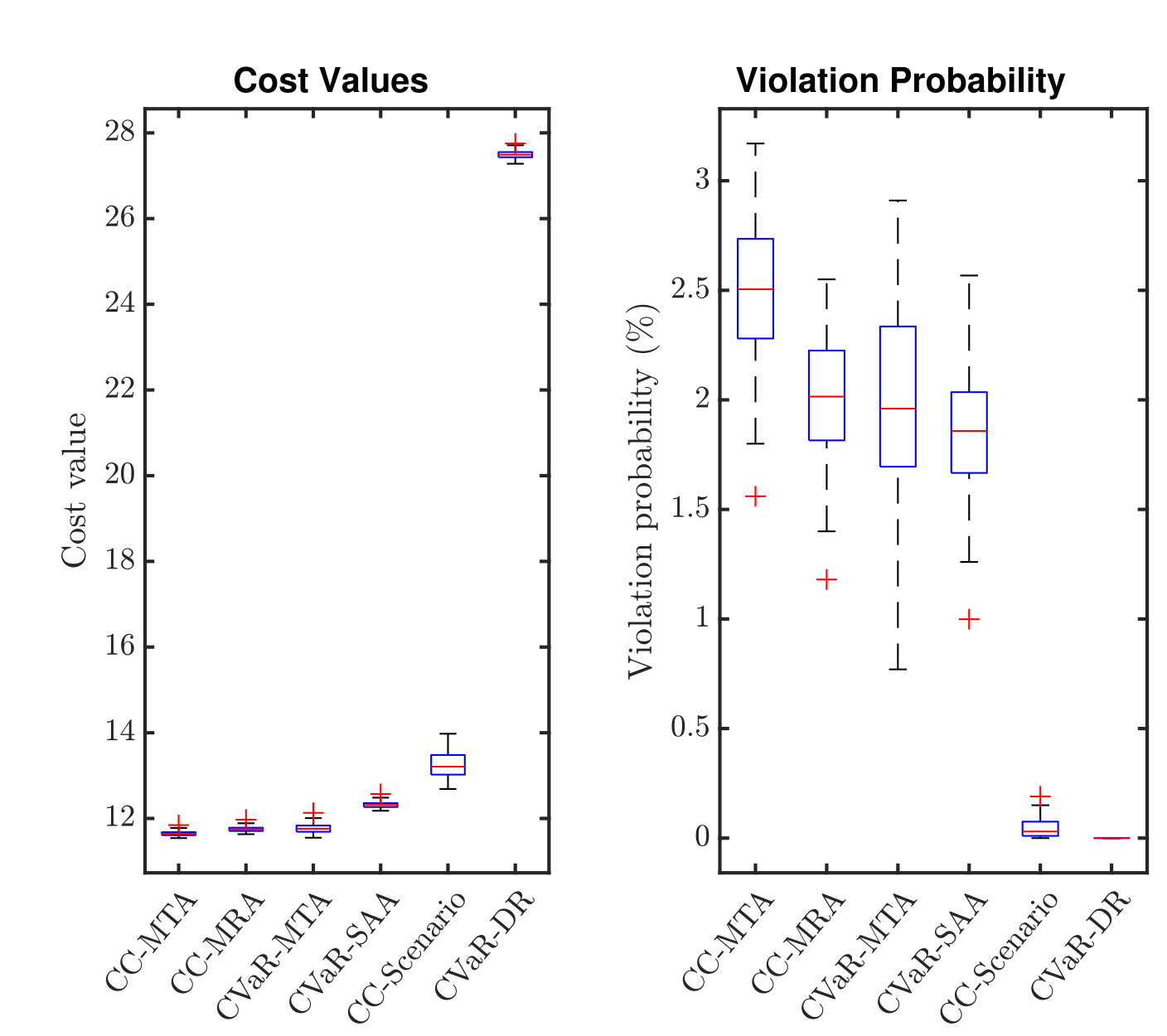}
    \caption{Optimality (left) and constraint violation (right) performances of the risk-constrained optimization methods.}
    \label{fig:GMMvsDR}
\end{figure}

The optimal solutions generated by different risk-constrained approaches and the empirical constraint violation probability over the 100 simulations are shown on the left and the right of Figure \ref{fig:GMMvsDR}, respectively. Based on the results, \ref{num:MTA} and \ref{num:MRA}, which rely on the assumption that the uncertainty conforms to GMM, achieve low costs yet have a large but feasible (i.e., less than or equal to $\epsilon = 5\%$) constraint violation probability. CVaR-constrained methods \ref{num:ICP} and \ref{num:SAA} yield more conservative solutions with a reduced probability of constraint violation. The CC-scenario method \ref{num:scenario} yields solutions that lead to a slightly higher cost than the CVaR-constrained framework. The CVaR-DR approach \ref{num:WassDR} generates much more conservative results than the other methods with a zero constraint violation probability.

The average computational time of the different risk-constrained methods over the 100 simulations is provided in the ``\textbf{Comp. time}" column of Table~\ref{tab:OLTPmethodsSummary}. The CVaR-constrained methods \ref{num:ICP} and \ref{num:SAA} significantly prolong the computation time. The CVaR-DR approach \ref{num:WassDR} in general takes longer than chance-constrained methods, whereas the CC-scenario method \ref{num:scenario} has the lowest computational time among all risk-constrained approaches. This trend is consistent with the observations when using these risk-constrained methods in open-loop trajectory planning problems.



\end{document}